\documentclass[10pt,fleqn]{article}
\usepackage{amsfonts}
\usepackage{amsmath}
\usepackage{color}
\usepackage{indentfirst,latexsym,bm}
\usepackage{amssymb}
\usepackage{amsmath}
\usepackage{graphics}
\usepackage{graphicx}

\numberwithin{equation}{section}

 \newcommand{\p}{\partial}
 \newcommand{\be}{\begin{eqnarray}}
 \newcommand{\ee}{\end{eqnarray}}
 \newcommand{\bee}{\begin{eqnarray*}}
 \newcommand{\eee}{\end{eqnarray*}}

\newtheorem{theorem}{Theorem}[section]

\newtheorem{definition}[theorem]{Definition}

\newtheorem{lemma}[theorem]{Lemma}

\newtheorem{proposition}[theorem]{Proposition}

\newtheorem{assumption}[theorem]{Assumption}
\newenvironment{proof}[1][Proof]{\textbf{#1.} }{\ \rule{0.5em}{0.5em}}
\begin{document}

\title{\Large\bf Indifference Pricing and Hedging
in a Multiple-Priors Model with Trading
Constraints\footnote{\noindent The work is supported by NNSF of
China (No.11271143,  11371155), University Special Research Fund for
Ph.D. Program (No. 20124407110001),
 NNSF of Zhejiang Province(No. Y6110775), and the Oxford-Man Institute of Quantitative Finance.
 We thank two anonymous referees for their valuable comments and helpful
suggestions on our paper.}}
\date{}
\author{
Huiwen Yan~\thanks{School of Mathematics and Computer Science,
  Guangdong University of Finance \& Economics, Guangzhou 510320, China,
  \texttt{hwyan10@gmail.com}}$\quad$
Gechun Liang~\thanks{Department of Mathematics, King's College
London, London WC2R 2LS, U.K.,
 \texttt{gechun.liang@kcl.ac.uk} }$\quad$
Zhou Yang~\thanks{School of Mathematical Science, South China Normal
 University, Guangzhou 510631, China,
 \texttt{yangzhou@scnu.edu.cn}}}
 \maketitle

\begin{abstract}
This paper considers utility indifference valuation of derivatives
under model uncertainty and trading constraints, where the utility
is formulated as an additive stochastic differential utility of both
intertemporal consumption and terminal wealth, and the uncertain
prospects are ranked according to a multiple-priors model of Chen
and Epstein (2002). The price is determined by two optimal
stochastic control problems  (mixed with optimal stopping time in
the case of American option) of forward-backward stochastic
differential equations. By means of backward stochastic differential
equation and partial differential equation methods,
%[Ambiguity, risk and asset return in continuous time, Econometrica, 70(2002), 1403-1443].
%In our setup where the investor is risk averse in her consumption,
we show that both bid and ask prices
%are independent of the investor's risk aversion, and
are closely related to the
Black-Scholes risk-neutral price with modified dividend rates. The
two prices will actually coincide with each other if there is no
trading constraint or the model uncertainty disappears. Finally, two
applications to European option and American
option are discussed.\\

%In this paper, a derivative pricing model in ambiguity market
% was established via indifferent pricing principle and the
% backward stochastic differential equation description for
% the worst case utility in ambiguity market in Chen and
% Epstein [Econometrica, 70(2002), 1403-1443]. Under some proper
% assumptions, it was proved that the buying price and the
% selling price in this model are same as the price in the
% standard Black-Scholes model if no investment restriction
% exists in the market. But if short selling is prohibited
% in the market, then the selling price or the buying price
% is different from the price in the B-S model, and it is
% governed by a forward-backward stochastic differential
% equation (for European derivative) or FBSDE with reflection (for
% American derivative). In Markovian frame, we transformed the FBSDE (RFBSDE)
% into PDE problems, and compare the prices of European derivatives or American
% derivatives in this model with the prices in B-S model by PDE method.


%{\emph{JEL Codes}}: G13,$\;$ C02,$\;$ C61\\

\end{abstract}

%%%%%%%%%%%%%%%%%%%%%%%%%%%%%%%%%%%%%%%%%%%%%%%%%%%%%%%%%%%
%%%%%%%%%%%%%%%%%%%%%%%%%%%%%%%%%%%%%%%%%%%%%%%%%%%%%%%%%%
\section{Introduction}

This paper considers derivative pricing in an incomplete financial
market with model uncertainty (ambiguity) and trading constraints.
Incompleteness means that investors are uncertain about the
risk-neutral probability measure which are used to price the
derivatives in the market. The investors rank the uncertain
prospects according to a multiple-priors model of Chen and Epstein
\cite{Chen}, where a continuous-time intertemporal version of
multiple-priors utility is formulated by using backward stochastic
differential equation (BSDE). The other source of incompleteness is
due to trading constraints such as short sale constraint. We take
the indifference pricing approach where the investor's utility is
formulated as an additive stochastic differential utility of both
intertemporal consumption and terminal wealth (see Duffie and
Epstein \cite{Duffie}).

The idea of indifference valuation is firstly introduced by Hodges
and Neuberger \cite{Hodges}, whereby the price for a derivative is
the cash amount that the investor is willing to pay such that she is
no worse off in expected utility terms than she would have been
without the derivative. The idea is further developed by Henderson
\cite{Henderson1}, Musiela and Zariphopoulou \cite{Musiela} and
Sicar and  Zariphopoulou \cite{SZ_bounds}, among others, under a
Markovian setting using the partial differential equation (PDE)
approach, and by Hu et al \cite{Hu}, Mania and Schweizer \cite{MS}
and Ankirchner et al \cite{Imkeller1} in a general non-Markovian
setting using the BSDE approach (see more references therein). On
the other hand, Becherer \cite{Becherer} and Davis \cite{Davis} use
the duality approach to study the indifference hedging strategy.
%For a general overview of indifference valuation, we refer to the recent monograph edited by Carmona \cite{MR2547456}, and especially the
%survey article by Henderson and Hobson \cite{Henderson-overview} therein.
However, most of the existing research is based on the assumption
that the investor is only concerned with her terminal wealth, and
ignores her intertemporal consumption and her risk aversion in the
consumption. It seems that the only exception is Cheridito and Hu
\cite{Cheridito}, which takes account of consumption on the top of
terminal wealth under the framework of Hu et al \cite{Hu}. In our
model, the investor takes account of not only the terminal wealth
but also the intertemporal consumption, and she is risk averse in
the consumption instead of the terminal wealth.

Different from most of utility indifference pricing models, where
the indifference prices are heavily distorted by the investor's risk
aversion in her terminal wealth, and therefore complicated in
general, the indifference prices in our model are independent of the
investor's risk aversion, and bear a striking resemblance to the
risk-neutral price. In fact, we show that both bid and ask prices in
our model are closely related to the risk-neutral price with
modified dividend rates. The deviation is actually caused by the
investor's uncertainty on the risk-neutral probability measure and
the existence of the trading constraints.

Model uncertainty is an important aspect in derivative pricing.
Indeed, uncertainty on the choice of an option pricing model can
lead to model risk in the valuation of portfolios of options, so one
must separate between risk (uncertainty on outcomes for which the
probabilities are known) and ambiguity (model uncertainty). In their
seminal work, Chen and Epstein \cite{Chen} introduce a
multiple-priors model under the framework of stochastic differential
utility. Cont \cite{Cont} introduces different risk measures to
quantity the model uncertainty. In the indifference pricing
framework, Jaimungal and Sigloch \cite{JS2009} introduce the concept
of robust indifference pricing (with the utility of terminal
wealth), which incorporates both risk aversion and model
uncertainty. They mainly use the idea from Anderson et al
\cite{AHS}, by modifying the optimization problem to maximize the
expected penalized utility of terminal wealth, while minimizing the
expected penalized utility over a set of equivalent measures. Since
the risk-neutral probability measure is naturally a dominant pricing
measure, we take the multiple-priors model from Chen and Epstein
\cite{Chen} in order to incorporate the model uncertainty, where the
probability measures in the set of priors are equivalent to the
risk-neutral dominant pricing measure.

The idea of applying Chen and Epstein's multiple-priors model to
derivative pricing is not new. For example, Guo et al \cite{Guo}
consider the pricing problem of exotic options, in particular
Parisian options, also under the multiple-priors framework of Chen
and Epstein \cite{Chen}. However, they use the idea of
super-replication rather than utility indifference valuation, so
they obtain pricing bounds rather than a price. On the other hand,
their concentration is more on numerical solutions of the pricing
bounds, which is different from our paper where we try to build up
an indifference pricing framework under model uncertainty. Moreover,
our model also includes the trading constraints on the top of model
uncertainty. Actually, these two factors result in the market
incompleteness in our model.

Although the model uncertainty causes the deviation of the bid and
ask prices from the risk-neutral price, we show that if there is no
trading constraint, then both bid and ask prices will coincide with
the risk-neutral price even under model uncertainty. In fact, if
there is no trading constraint, the investor can invest arbitrary
position in the underlying stocks to hedge the model uncertainty in
our indifference pricing setup. On the other hand, if there exist
trading constraints, then the prices will deviate from the
risk-neutral price. For example, if the payoff of an option is
monotone with respect to the prices of the underlying, the investor
needs to hold an opposite position in order to hedge her risk
exposure. Trading constraints such as short sale constraint will
prohibit her arbitrary position, and therefore result in prices
different from the risk-neutral price. We should remark that the
price deviation is not only caused by trading constraints but also
model uncertainty. These two factors tangling together impact the
bid and ask prices, in both European option and American option
cases. We also obtain the convergence rate of the indifference
prices to the risk-neutral price as the model uncertainty disappears
(see Proposition 3.5 for the case of European option and Proposition
4.8 for the case of American option).

If the investor is allowed to exercise her option at any time before
maturity as in the American option setting, the definition of the
corresponding indifference price needs to be modified. In such a
situation, the investor needs to compare two optimal investment
problems with different time horizons, and she chooses not only her
optimal trading strategy but also her optimal exercise time. By the
time consistency, we derive an intertemporal wealth which consists
of not only the usual wealth , but also the value of the remaining
optimal consumption (see Definition \ref{definitionofAmerican}). In
a Markovian setting, we also effectively use the
Alexandrov-Bakel'man-Pucci (A-B-P) comparison principle for
variational inequalities to deduce various properties of the bid and
ask prices for American option.

{The utility indifference pricing models are usually formulated as
two optimal portfolio problems (see \eqref{Pricing3} and
\eqref{Pricing4}).  Concretely speaking, the price $P$ is determined
from an equation $F(\cdot,P(\cdot))=G(\cdot)$, and functions $F$ and
$G$ where the value functions of two optimal stochastic control
problems. There are two methods to solve these problems, such as the
BSDE approach and the PDE approach. The BSDE approach is based on
martingale optimality principle (see \cite{Hu,MS}) or the
risk-sensitive control (see a recent work by Henderson and Liang
\cite{Henderson-liang}), and the price is expressed via the
solutions of two BSDEs. The PDE method is based on dynamic
programming principle, and the price is described via the solutions
of two HJB equations (see \cite{Henderson1,Musiela,SZ_bounds}). Due
to the complexities of BSDEs and HJB equations, it is usually
difficult to investigate the properties of the utility indifference
prices.}

{In our model, the price is still determined by the equation
$F(\cdot,P(\cdot))=G(\cdot)$, where $F$ and $G$ are the value
functions of the two optimal control problems of forward-backward
stochastic equations (FBSDEs). In the American option case, the
optimal control problems involve optimal stopping time problems, and
the optimal strategy consists of optimal consumption, optimal
investment and optimal stopping time.}

In order to completely solve the problems, we firstly analyze the
two optimal portfolio problems and express the price as the solution
of nonlinear BSDEs via the following idea. We define an indirect
utility by subtracting the wealth (and the contingent claim) from
the original stochastic differential utility (SDU) to represent the
indifference price as the solution of BSDE with Lipschitz continuous
driver. Mathematically speaking, we translate the original
stochastic control problem of FBSDE into finding a maximum solution
to a family of BSDEs with different drivers. Then by the comparison
principle for BSDE, we find the nonlinear BSDE for the optimal
solution and express the price via the corresponding FBSDE.

By applying the  Feynman-Kac formula, we express the prices as the
solutions of quasi-linear PDE (for European option) or variational
inequality (VI, for American option). Then via the method of PDE, we
improve the regularity of the value function, and analyze the
properties of the price, the optimal investment, consumption and
stopping strategies. In this paper, We give a general and technical
proof of improving the regularity of the solution of PDE or VI under
low regularity of terminal value and obstacles. Moreover, we prove
that PDE is a special case of VI if we choose a proper obstacle
under some general assumptions. Thanks to the improved regularity,
we achieve some concise results as mentioned above. With these
results, it is easy to calculate the prices or theoretically analyze
their properties by the standard method. Moreover, we use some PDE
estimates to show the convergence result of the corresponding prices
when the ambiguity market shrinks to the standard market.

The paper is organized as follows: We present our indifference
pricing and hedging model in Section 2, and apply it to two concrete
settings, namely European option and American option in Section 3
and Section 4 respectively. Some further technical details on the
results of relevant PDEs are provided in the Appendix.

%%%%%%%%%%%%%%%%%%%%%%%%%%%%%%%%%%%%%%%%%%%%%%%%%%%%%%%%5
%%%%%%%%%%%%%%%%%%%%%%%%%%%%%%%%%%%%%%%%%%%%%%%%%%%%%%%%%

\section{Indifference Pricing and Hedging Model}

For a fixed time horizon $T>0$, let $W=(W^1,\cdots,W^n)^T$ be an
$n$-dimensional Brownian motion on a filtered probability space
$(\Omega,\mathcal{F}, \mathbb{F}=\{{\mathcal{F}}_t\} ,\mathbb{P})$
satisfying the \emph{usual conditions}, where $\mathbb{F}$ is the
augmented filtration generated by the Brownian motion $W$, and
$\mathbb{P}$ is interpreted as the risk-neutral probability measure.
Herein the superscript $T$ denotes the matrix transposition. The
market consists of one risk-free asset $B$ with the risk-free
interest rate $r(\cdot)$, and $n$ risky assets
$S=(S^1,\cdots,S^n)^T$, whose price processes under the risk-neutral
probability measure $\mathbb{P}$ are given by
\begin{equation}\label{stateequation}
 S^i_{s}=S^i_{t}+\int^s_t r(u)S^i_u\,du+\sum\limits_{j=1}^n\int^s_t \sigma_{i\,j}(u)S^i_u\,dW^j_u,
\end{equation}
for $1\leq i\leq n$ and $0\leq t\leq s\leq T$, where
$\sigma(\cdot)=(\sigma_{ij}(\cdot))_{1\leq i,j\leq n}$ is the
volatility matrix. The coefficients satisfy the following
assumption:

\begin{assumption}\label{Assumption1}
The risk-free interest rate $r(\cdot)$ and the volatility matrix
$\sigma(\cdot)$ are continuous functions, and the
volatility matrix $\sigma(\cdot)$ is positive definite.
\end{assumption}

However, investors are uncertain about the risk-neural probability
measure $\mathbb{P}$, and rank the uncertain prospects according to
a multiple-priors model, which was initially proposed by Chen and
Epstein \cite{Chen}. They represent the set $\Theta$ of priors on
$(\Omega,\mathcal{F}_{T})$ by the set of probability measures
equivalent to $\mathbb{P}$:
$$
 \Theta\triangleq\Bigg\{\mathbb{Q}:{d\mathbb{Q}\over d\mathbb{P}}
 =\exp\left(-{1\over 2}\int_0^{T}|\xi_s|^2ds-\int_0^{T}(\xi_s)^TdW_s\right)
\Bigg\},
$$
for $\xi=(\xi^1,\cdots,\xi^n)^T\in\Xi$, where $\Xi$ is the set of
$\mathbb{F}$-adapted process valued in a compact and convex subset
$\mathcal{O}\subset\mathbb{R}^n$ including the origin $0$. More
generally, the density on $(\Omega,\mathcal{F}_t)$ is defined as
$\mathbb{E}_{\mathbb{P}}[\frac{d\mathbb{Q}}{d\mathbb{P}}|\mathcal{F}_t]$.
Hence, $\Theta$ is indeed the set of equivalent probability measures
which includes the risk-neutral probability measure $\mathbb{P}$ as
the dominant pricing measure.

For any starting time $t\in[0,T]$, a representative investor makes
inter-temporal consumption and invests the remaining wealth in the
risk-free asset and the risky assets in the remaining time interval
$[t,T]$, resulting in her wealth equation:
\begin{align}\label{SDE}
 X_s^{X_t;\,\pi,c}&=X_t+\int_t^s\frac{X_u^{X_t;\pi,c}-\sum_{i=1}^n\pi^i_u}{B_u}dB_u+
 \sum_{i=1}^{n}\int_t^s\frac{\pi^i_u}{S^i_u}dS^i_u-\int_t^sc_udu \nonumber\\
 &=X_t+\int_t^s \Big[\,r(u)X_u^{X_t;\,\pi,c}-c_u\,\Big]du
 +\int_t^s(\pi_u)^T\sigma(u)dW_u,
\end{align}
where $(\pi,c)$ is the portfolio-consumption strategy with $c$ being
the inter-temporal consumption rate, and
$\pi=(\pi^1,\cdots,\pi^n)^T$ being the amount of money invested in
the risky assets $S=(S^1,\cdots,S^n)^T$, both in the admissible set
$\Pi[t, T]$:
$$
\Pi[t,T]\triangleq\{(\pi,c):\pi\in
L^2_{\mathbb{F}}(t,T;\mathcal{A}),\,
 c \in L^2_{\mathbb{F}}(t,T;\mathbb{R}_+)\}
$$with
$$
 L^2_{\mathbb{F}}(t,T;\mathcal{A})\triangleq
 \left\{\,\pi:\mathbb{F}\mbox{-adapted}, \mbox{valued in}\,\mathcal{A},\
\mbox{and} \
 \mathbb{E}_\mathbb{P}\left[\,\int_t^T
 |\pi_s|^2ds\,\right]<\infty\,\right\},
$$
where $\mathcal{A}$ is a closed subset of $\mathbb{R}^n$.

 Note that the investor in fact makes her decision under
$\mathbb{Q}$ but not under $\mathbb{P}$, as she is uncertain about
the risk-neutral probability $\mathbb{P}$. By Girsanov's theorem,
$\overline{W}=(\overline{W}^1,\cdots,\overline{W}^n)^T$ with
$\overline{W}^j=W^j+\int_0^{\cdot}\xi_s^jds$ is the Brownian motion
under $\mathbb{Q}$, and the investor's wealth equation (\ref{SDE})
under $\mathbb{Q}$ is
\begin{align}\label{SDE1}
 X_s^{X_t;\,\pi,c}
 &=X_t+\int_t^s \Big[\,r(u)X_u^{X_t;\,\pi,c}-c_u-\sum_{i,j=1}^n\sigma_{ij}(u)\pi^i_u\xi^j_{u}\,\Big]du
 +\sum\limits_{i,j=1}^n\int_t^s \sigma_{ij}(u)\pi^i_u d\overline{W}^j_u\nonumber\\
 &=X_t+\int_t^s \Big[\,r(u)X_u^{X_t;\,\pi,c}-c_u-(\pi_u)^T\sigma(u)\xi_{u}\,\Big]du
 +\int_t^s (\pi_u)^T\sigma(u)d\overline{W}_u.
\end{align}
Hence, the portfolio-consumption strategy $(\pi,c)$ does impact the
investor's wealth $X^{X_t;\pi,c}$ through its drift terms.

The investor has an additive utility defined on the priors set
$\Theta$, which is formulated as a stochastic differential utility
as in Duffie and Epstein \cite{Duffie}:
$$
U^\mathbb{Q}_t\triangleq\mathbb{E}_\mathbb{Q}
 \Bigg[\,\int_t^\tau \left(v(u,c_u)-r(u)U^\mathbb{Q}_u\right)\,du+U^\mathbb{Q}_\tau\,\Bigg|\,{\cal F}_t\,\Bigg],
$$ with the bequest at the terminal time $T$: $U_T^{\mathbb{Q}}=X_T^{X_t;\pi,c}$, where
$\tau\in{\cal U}[t,T]$ is the set of $\mathbb{F}$-stopping times
valued in $[t,T]$, $v(\cdot,\cdot)$ is the time-dependent utility of
the inter-temporal consumption rate $c$, and $X^{X_t;\pi,c}$ is the
wealth process. The investor chooses the worst scenario from the
priors set $\Theta$ as her utility:
\begin{equation}\label{SDU}
U_t=\inf_{\mathbb{Q}\in\Theta}U_t^{\mathbb{Q}}.
\end{equation}

The utility $v(\cdot,\cdot)$ of the consumption rate $c$ satisfies
the following assumption:
\begin{assumption}\label{Assumption2}
 For any $t\in[0,T]$, the time-dependent utility
 $v:[\,0,T\,]\times\mathbb{R}\mapsto [-\infty,\infty)$ is
(1) concave, nondecreasing and upper semi-continuous; (2) the
half-line $dom_t(v)\triangleq\{x\in[\,0,+\infty):v(t,x)>-\infty\}$
is a nonempty subset of $[\,0,\infty)$; and (3) $\partial_x
v(t,\cdot)$ is continuous, positive, and strictly
 decreasing in the interior of $dom_t(v)$, and
$$
 \lim\limits_{x\rightarrow+\infty}\sup\limits_{t\in[\,0,T\,]}\partial_x v(t,x)=0.
$$
\end{assumption}

A typical example of $v(\cdot,\cdot)$ is the power utility which is
penalized to $-\infty$ when the consumption rate is negative. By
Section 3.4 of Karatzas and Shreve \cite{Karatzas}, there exists
$\widehat{c}(t)$ such that
\begin{equation}\label{optimalconsumption}
v^*(t)\triangleq
v(t,\widehat{c}(t))-\widehat{c}(t)=\sup_{x\in\mathbb{R}}\{v(t,x)-x\}.
\end{equation}
That is, $v^*(t)$ is the convex dual of $v(t,\cdot)$ at level $1$
for $t\in[0,T]$.

The investor values a contingent claim in the market, whose payoff
is an $\mathcal{F}_{T}$-measurable random variable $\xi\in
L^2_{{\cal F}_T}(\mathbb{R}_+)$:
$$
 L^2_{{\cal F}_T}(\mathbb{R}_+)\triangleq
 \left\{\,\xi:\mathcal{F}_{T}\text{-measurable,\ valued\ in}\
 \mathbb{R}_+,\ \text{and}\
 \mathbb{E}_\mathbb{P}\left[\,|\xi|^2\,\right]<\infty\,\right\}.
$$

If there is no ambiguity about the risk-neutral probability measure
$\mathbb{P}$, no inter-temporal consumption and no trading
constraints, it is known that the risk-neutral price process $D$ of
this contingent claim and the corresponding hedging strategy
$Y=(Y^1,\cdots,Y^n)^T$ (normalized by the volatility matrix
$\sigma(\cdot)$) are the unique solutions to the following linear
BSDE:
 \begin{align}\label{contingentbsde}
 D_t=\xi-\int_t^T r(u)D_udu-\int_t^T
 (Y_u)^TdW_u=\mathbb{E}_{\mathbb{P}}\left[e^{-\int_t^{T}r(u)du}\xi|\mathcal{F}_{t}\right]
 \end{align}
for $(D,Y)\in L^2_{\mathbb{F}}(0,T;\mathbb{R})\times
L^2_{\mathbb{F}}(0,T;\mathbb{R}^n)$. See Section 1 of El Karoui et
al \cite{ElKaroui} for the further details. However, due to the
ambiguity, the consumption and the trading constraints,
(\ref{contingentbsde}) is not valid.

We consider utility indifference valuation for such a contingent
claim with the payoff $\xi$.
%For a general overview of utility indifference valuation, we refer
%to the recent monograph edited by Carmona \cite{MR2547456}, and
%especially the survey article by Henderson and Hobson
%\cite{Henderson-overview} therein.
We write the utility $U_t$ as $U_t(X_T^{X_t;\pi,c})$, when we want
to emphasize the dependence of the utility $U$ on the bequest
$X_T^{X_t;\pi,c}$.
\begin{definition}\label{def}
The bid price $P^b(X_t;\xi)$ and the ask price $P^s(X_t;\xi)$ of the
contingent claim with the payoff $\xi$ are defined implicitly by the
requirement that
\begin{align}\label{Pricing3}
 \mathop{{\rm ess.sup}}\limits_{(\pi,c)\in\Pi[t,T]}U_t(X^{X_t;\,\pi,\,c}_T)
 &=\mathop{{\rm ess.sup}}\limits_{(\pi,c)\in\Pi[t,T]}U_t(\xi+X^{X_t-P^b(X_t;\xi);\,\pi,\,c}_T);
 \\[2mm]\label{Pricing4}
 \mathop{{\rm ess.sup}}\limits_{(\pi,c)\in\Pi[t,T]}U_t(X^{X_t;\,\pi,\,c}_T)
 &=\mathop{{\rm
 ess.sup}}\limits_{(\pi,c)\in\Pi[t,T]}U_t(-\xi+X^{X_t+P^s(X_t;\xi);\,\pi,\,c}_T),
\end{align}
where $X^{X_t;\pi,c}_{\cdot}$ follows the wealth equation
(\ref{SDE}) starting from time $t$ under $\mathbb{P}$, or
equivalently (\ref{SDE1}) under $\mathbb{Q}$, and $U_t(\cdot)$ is
the worst case of the stochastic differential utility at time $t$
defined by (\ref{SDU}), i.e. $U_t=U_t(X^{X_t;\,\pi,\,c}_T)$
satisfies
$$
U_t=\inf_{\mathbb{Q}\in\Theta}U_t^{\mathbb{Q}}
=\inf\limits_{\mathbb{Q}\in\Theta}\,\mathbb{E}_\mathbb{Q}
 \Bigg[\,\int_t^T \left(v(u,c_u)-r(u)U^\mathbb{Q}_u\right)\,du+X^{X_t;\,\pi,\,c}_T\,\Bigg|\,{\cal
 F}_t\,\Bigg].
$$
\end{definition}

In terms of utility maximization, we are thus indifferent between
buying or not buying the contingent claim with the payoff $\xi$ for
the bid price $P^b(X_t;\xi)$, and indifferent between selling or not
selling the contingent claim with the payoff $\xi$ for the ask price
$P^s(X_t;\xi)$, while the utility is chosen as the worst scenario
from the priors set $\Theta$.

Our model deviates from the existing literature in the following
three folds: (1) The utility is formulated in terms of a stochastic
differential utility of both the inter-temporal consumption and the
terminal wealth, while the most of exiting literature only considers
the expected utility of the terminal wealth; (2) The model
uncertainty is taken into account in a dynamic consistent way by
employing the multiple-priors model of Chen and Epstein \cite{Chen};
(3) The trading constraints are also considered in our indifference
valuation model. We shall see (2) and (3) lead to some new and
interesting features of the indifference price.

Our main tool to characterize the indifference price is the theory
of BSDE. By Theorem 2.2 of Chen and Epstein \cite{Chen}, the utility
$U$ is represented as the unique solution of the following BSDE:
\begin{equation}\label{utility}
 U_t=X_T^{X_t;\pi,c}
 +\int_t^T \left[\,v(u,c_u)-r(u)U_u-
 \max\limits_{\xi\in\Xi}\sum\limits_{j=1}^n\xi^j_{u}Z_u^j\,\right]\,du
 -\sum\limits_{j=1}^n\int_t^T Z_u^jdW^j_u
\end{equation}
for $(U,Z)\in L^2_{\mathbb{F}}(0,T;\mathbb{R})\times
L^2_{\mathbb{F}}(0,T;\mathbb{R}^n)$. The BSDEs in this paper are
always considered in the above space.

Note that the maximum term in the above bracket is in fact pathwise
maximum:
$$\left(\max\limits_{\xi\in\Xi}\sum\limits_{j=1}^n\xi^j_{t}Z_t^j\right)(\omega)=
\max_{\xi_t(\omega)\in\mathcal{O}}\sum_{j=1}^n\xi^j_t(\omega)Z_t^j(\omega)$$
for any $t\in[0,T]$ and $\omega\in\Omega$, so we will write the
above two maximization problems synonymously.\\

Our first main result is the following representation result for the
bid price and the ask price.

\begin{theorem}\label{theorem} Suppose that Assumptions \ref{Assumption1} and
\ref{Assumption2} are satisfied. Then the bid price $P^{b}(X_t;\xi)$
and the ask price $P^{s}(X_t;\xi)$ are uniquely determined by
Definition \ref{def}, and are both independent of the initial wealth
$X_t$. They are denoted as $P^{b}(t;\xi)$ and $P^{s}(t;\xi)$
respectively, and have the representations:
\begin{align}\label{estimate1}
 &&P^b(t;\xi)=\mathop{{\rm ess.sup}}\limits_{(\pi,c)\in\Pi[t,T]}U_t(\xi+X^{0;\,\pi,c}_T)
 -{\cal R}(t)\leq D_t ,
 \\[2mm]\label{estimate2}
 &&P^s(t;\xi)=-\mathop{{\rm ess.sup}}\limits_{(\pi,c)\in\Pi[t,T]}U_t(-\xi+X^{0;\,\pi,c}_T)
 +{\cal R}(t)\geq D_t,
\end{align}
where $\mathcal{R}(\cdot)$ is the value of the optimal consumption:
$$\mathcal{R}(t)=\int_t^{T}e^{-\int_{t}^sr(u)du}v^*(s)ds.$$
\end{theorem}

\begin{proof} We only consider the case of the bid price, while the case of the ask price
is similar. Note that the utility maximization problem on the left hand side (LHS) of
(\ref{Pricing3}) is a special case of the one on the right hand side (RHS) with $\xi=0$.

We first show that the solutions of both utility maximization
problems in (\ref{Pricing3}) exist, i.e.
$$\mathop{{\rm ess.sup}}\limits_{(\pi,c)\in\Pi[t,T]}U_t(\xi+X^{X_t;\pi,c}_T)<+\infty$$ for any payoff $\xi\in
L^2_{\mathcal{F}_{T}}(\mathbb{R}_+)$ and any initial wealth $X_t\in
L^2_{\mathcal{F}_{t}}(\mathbb{R})$. Indeed, if we define an indirect
utility $\widehat{U}$ by subtracting the contingent claim $D$ and
the wealth $X^{X_t;\pi,c}$ from the original utility $U$ as
$$
\widehat{U}_{s}=U_{s}(\xi+X^{X_t;\pi,c}_T)-(D_{s}+X_{s}^{X_t;\pi,c});\qquad
\widehat{Z}_s^j=Z_s^j-\left(
Y_s^j+\sum_{i=1}^n\sigma_{ij}(s)\pi^i_{s}\right)
$$
for $s\in[t,T]$, by (\ref{SDE}), (\ref{contingentbsde}) and
(\ref{utility}), it is easy to verify that
$(\widehat{U},\widehat{Z})$ satisfies the following BSDE:
\begin{equation}\label{wealthbsde}
 \widehat{U}_t=\int_t^T \left\{\,-r(u)\widehat{U}_u+[\,v(u,c_u)-c_u\,]
 -\max\limits_{\xi\in\Xi}\sum\limits_{j=1}^n\xi_u^jZ_u^j\,\right\}\,du
 -\sum\limits_{j=1}^n\int_t^T \widehat{Z}_u^jdW^j_u.
\end{equation}
Note that
$$
v(u,c_u)-c_u\leq v^*(u);\qquad
 -\max\limits_{\xi\in\Xi}\sum\limits_{j=1}^n\xi_u^jZ_u^j\leq0.
$$
By the BSDE comparison theorem (see Theorem 2.2 of \cite{ElKaroui}),
$\widehat{U}_t\leq \mathcal{R}(t)$, where $\mathcal{R}$ is given by
\begin{equation}\label{wealthbsde0}
 \mathcal{R}(t)=\int_t^T \left[\,-r(u)\mathcal{R}(u)+v^*(u)\,\right]\,du
 -\int_t^T (\mathcal{Q}(u))^TdW_u,
\end{equation}
which has a unique solution in
$L^2_{\mathbb{F}}(0,T;\mathbb{R})\times
L^2_{\mathbb{F}}(0,T;\mathbb{R}^n)$:
$$(\mathcal{R}(t),\mathcal{Q}(t))=\left(\int_t^{T}e^{-\int_{t}^sr(u)du}v^*(s)ds,\;0\right).$$
Hence $U_t(\xi+X^{X_t;\,\pi,c}_T)\leq D_t+X_t+\mathcal{R}(t)$ for
any $(\pi,c)\in\Pi$, and we have proved the upper bound:
\begin{equation}\label{bsdeupperbound}
 \mathop{{\rm ess.sup}}\limits_{(\pi,c)\in\Pi[t,T]}U_t(\xi+X^{X_t;\,\pi,c}_T)\leq D_t+X_t
 +\mathcal{R}(t)<\infty.\quad
\end{equation}

Next, we show the representation (\ref{estimate1}).  We first solve
the utility maximization problem on LHS of (\ref{Pricing3})
explicitly. By taking $(\pi,c)=(0,\widehat{c})$ in (\ref{utility}),
we have
$$
 \widehat{U}_t=\int_t^T \left\{\,-r(u)\widehat{U}_u+v^*(u)
 -\max\limits_{\xi\in\Xi}\sum\limits_{j=1}^n\xi_u^jZ_u^j\,\right\}\,du
 -\sum\limits_{j=1}^n\int_t^T Z_u^jdW^j_u,
$$
where $\widehat{U}_s\triangleq
U_s(X_T^{X_t;0,\hat{c}})-X_s^{X_t;0,\widehat{c}}$ for $s\in[t,T]$.
The above BSDE has a unique solution in
$L^2_{\mathbb{F}}(0,T;\mathbb{R})\times
L^2_{\mathbb{F}}(0,T;\mathbb{R}^n)$, which is the same as
(\ref{wealthbsde0}):
$(\widehat{U}_t,Z_t)=(\mathcal{R}(t),\mathcal{Q}(t))$. Hence
$$\mathop{{\rm ess.sup}}\limits_{(\pi,c)\in\Pi[t,T]}U_t(X^{X_t;\,\pi,c}_T)
\geq
U_t(X^{X_t;\,0,\widehat{c}}_T)=\widehat{U}_t+X_t=\mathcal{R}(t)+X_t.$$
The above inequality is actually the equality. Indeed, By taking
$\xi=0$ in (\ref{bsdeupperbound}), we have the upper bound:
$$\mathop{{\rm ess.sup}}\limits_{(\pi,c)\in\Pi[t,T]}U_t(X^{X_t;\,\pi,c}_T)
\leq \mathcal{R}(t)+X_t,$$ so we have proved that
\begin{equation}\label{LHS}
 \mathop{{\rm
 ess.sup}}\limits_{(\pi,c)\in\Pi[t,T]}U_t(X^{X_t;\,\pi,c}_T)=
\mathcal{R}(t)+X_t
\end{equation}
with the optimal portfolio-consumption strategy
$(\pi^*,c^*)=(0,\widehat{c}\,)$ for the utility maximization problem
on LHS of (\ref{Pricing3}).

On the other hand, note that the utility maximization problem on RHS
of (\ref{Pricing3}) is independent of the initial wealth:
\be\label{RHS}
 \mathop{{\rm ess.sup}}\limits_{(\pi,c)\in\Pi[t,T]}U_t(\xi+X^{X_t;\,\pi,c}_T)
 %=\mathop{{\rm
% ess.sup}}\limits_{(\pi,c)\in\Pi[t,T]}U_t(\xi+X^{0;\,\pi,c}_T+X_t)
 =\mathop{{\rm
 ess.sup}}\limits_{(\pi,c)\in\Pi[t,T]}U_t(\xi+X^{0;\,\pi,c}_T)+X_t.\quad
\ee

By combining (\ref{LHS}) and (\ref{RHS}), we obtain the
representation (\ref{estimate1}), which also shows that the bid
price $P^b(t;\xi)$ is independent of the initial wealth $X_t$.
Finally, by taking $X_t=0$ in the inequality (\ref{bsdeupperbound})
and using the representation
 (\ref{estimate1}), we obtain the upper bound of the bid price:
$$P^b(t;\xi)=\mathop{{\rm ess.sup}}\limits_{(\pi,c)\in\Pi[t,T]}U_t(\xi+X^{0;\,\pi,c}_T)
 -{\cal R}(t)\leq D_t+\mathcal{R}_t-\mathcal{R}_t=D_t.$$
\end{proof}

Our next result further characterizes the indifference prices in
terms of the solutions to BSDEs. We show that different admissible
sets $\Pi$ result in different bid prices $P^b(t;\xi)$ and ask
prices $P^s(t;\xi)$, and in particular, if there is no constraint in
$\Pi$, then both the bid price and the ask price coincide with the
risk-neutral price even with model ambiguity.

For $t\in[0,T]$ and $\omega\in\Omega$, we define the subset
$\mathcal{B}_t(\omega)\subset\mathbb{R}^n$ by
$$
 {\cal B}_t(\omega)\triangleq\left\{\,z\in\mathbb{R}^n: z_j=\sum_{i=1}^n\sigma_{ij}(t)\pi^{i}_t(\omega)
 \;\mbox{with}\;\pi_t(\omega)\in {\cal A}\,\right\}.
$$
Note that $\mathcal{B}_t(\omega)$ is still closed since
$\mathcal{A}$ is closed and $\sigma_{ij}(\cdot)$ is bounded. For any
$\overline{z}\in\mathbb{R}^n$, we further introduce
$$d_{\mathcal{O}}(\overline{z},\mathcal{B}_t(\omega))=
\min_{z\in\mathcal{B}_t(\omega)}\max_{\xi_t(\omega)\in\mathcal{O}}\sum_{j=1}^n\xi^j_t(\omega)(\overline{z}_j+z_j),
$$
where $\mathcal{O}$ is a closed and convex subset of $\mathbb{R}^n$
including the origin $0$. Since $\mathcal{B}_t(\omega)$ is closed,
there exists at least one  point in $\mathcal{B}_t(\omega)$ which
minimizes the support function $\delta_{\mathcal{O}}(\cdot)$ of the
compact and convex set $\mathcal{O}$:
$$\delta_{\mathcal{O}}(\overline{z},z)=\max_{\xi_t(\omega)\in\mathcal{O}}\sum_{j=1}^n\xi^j_t(\omega)(\overline{z}_j+z_j),$$
and we denote such minimal point as ${\rm
argmin}(\overline{z},\mathcal{B}_t(\omega))$. In the following, we
will omit $\omega$ if no confusion may arise.

\begin{theorem}\label{theorem2}
Suppose that Assumptions \ref{Assumption1} and \ref{Assumption2} are
satisfied. Then the bid price $P^b(t;\xi)=U_t^b$ and the ask price
$P^s(t;\xi)=U^s_t$, where $U^b$ and $U^s$ are the unique solutions
to the following BSDEs respectively:
\begin{align}\label{Pricing5}
 U^b_t&=\xi
 +\int_t^T \Big[\,-r(u)U^b_u-d_{\mathcal{O}}\,(Z^b_u,{\cal B}_u)\,\Big]\,du
 -\int_t^T (Z^b_u)^{T}dW_u;
 \\[2mm]\label{Pricing6}
 U^s_t&=\xi
 +\int_t^T \Big[\,-r(u)U^s_u+d_{\mathcal{O}}\,(-Z^s_u,{\cal B}_u)\,\Big]\,du
 -\int_t^T (Z^s_u)^TdW_u
\end{align}
for $(U^b,Z^b)\in L^2_{\mathbb{F}}(0,T;\mathbb{R})\times
L^2_{\mathbb{F}}(0,T;\mathbb{R}^n)$ and $(U^s,Z^s)\in
L^2_{\mathbb{F}}(0,T;\mathbb{R})\times
L^2_{\mathbb{F}}(0,T;\mathbb{R}^n)$. The optimal
portfolio-consumption strategy for the bid price is
$$(\pi^*,c^*)=((\sigma^{T})^{-1}{\rm argmin}(Z^b,{\cal B}),
\widehat{c}),$$ and for the ask price is
$$(\pi^*,c^*)=((\sigma^{T})^{-1}{\rm argmin}(-Z^s,{\cal B}),
\widehat{c}).$$

Moreover, if $\mathcal{A}=\mathbb{R}^n$, i.e. there is no trading
constraint, then both the bid price $P^b(t;\xi)$ and the ask price
$P^s(t;\xi)$ coincide with the risk neutral price $D_t$:
$$P^b(t;\xi)=P^s(t;\xi)=D_t=\mathbb{E}_{\mathbb{P}}\left[e^{-\int_t^{T}r(u)du}\xi|\mathcal{F}_{t}\right].$$
The optimal portfolio-consumption strategy for the bid price is
$(\pi^*,c^*)=(-(\sigma^{T})^{-1}Y, \widehat{c}\,)$, and for the ask
price is $(\pi^*,c^*)=((\sigma^{T})^{-1}Y, \widehat{c}\,)$, where
$Y$ is given by (\ref{contingentbsde}).

%\begin{itemize}
%\item If the admissible set $\Pi=\widehat{\Pi}$, then both the bid price $P^b(t;\xi)$ and the ask price $P^s(t;\xi)$
%coincide with the risk neutral price $D_t$:
%$$P^b(t;\xi)=P^s(t;\xi)=D_t=\mathbb{E}_{\mathbb{P}}\left[e^{\int_t^{T}r(u)du}\xi|\mathcal{F}_{t}\right].$$
%The optimal portfolio-consumption strategy for the bid price is
%$(\pi^*,c^*)=(-(\sigma^{T})^{-1}Y, \widehat{c})$, and for the ask
%price is $(\pi^*,c^*)=((\sigma^{T})^{-1}Y, \widehat{c})$, where $Y$
%is given by (\ref{contingentbsde}).
%
%\item If the admissible set $$\Pi=\widehat{\Pi}\cap\left\{\pi:
%\pi_t(\omega)\in \mathcal{A}\ \text{for}\ t\in[0,T]\ \text{and}\
%\omega\in\Omega\right\},$$ then the bid price $P^b(t;\xi)=U_t^b$ and
%the ask price $P^s(t;\xi)=U^s_b$, where
%\begin{align}\label{Pricing5}
% U^b_t&=\xi
% +\int_t^T \Big[\,-r(u)U^b_u-d_{\mathcal{O}}\,(Z^b_u,{\cal B}_u)\,\Big]\,du
% -\int_t^T (Z^b_u)^{T}dW_u;
% \\[2mm]\label{Pricing6}
% U^s_t&=\xi
% +\int_t^T \Big[\,-r(u)U^s_u+d_{\mathcal{O}}\,(-Z^s_u,{\cal B}_u)\,\Big]\,du
% -\int_t^T (Z^s_u)^TdW_u.
%\end{align}
%The optimal portfolio-consumption strategy for the bid price is
%$(\pi^*,c^*)=((\sigma^{T})^{-1}{\rm argmin}(Z^b,{\cal B}),
%\widehat{c})$, and for the ask price is
%$(\pi^*,c^*)=((\sigma^{T})^{-1}{\rm argmin}(-Z^s,{\cal B}),
%\widehat{c})$.
%\end{itemize}
\end{theorem}

\begin{proof} We again only consider the case of the bid price, as the case of the ask price
is similar. By the representation (\ref{estimate1}), we only need to
solve $$\mathop{{\rm
ess.sup}}\limits_{(\pi,c)\in\Pi[t,T]}U_t(\xi+X^{0;\,\pi,c}_T).$$

For the initial wealth $X_t=0$, define the indirect utility
$\overline{U}$ by subtracting the wealth $X^{0;\pi,c}$ from the
original utility $U$ as
$$
\overline{U}_{s}=U_{s}(\xi+X^{0;\pi,c}_T)-X_{s}^{0;\pi,c};\qquad
\overline{Z}_s^j=Z_s^j-\sum_{i=1}^n\sigma_{ij}(s)\pi^i_{s}
$$
for $s\in[t,T]$. By (\ref{SDE1}) and (\ref{utility}), it is easy to
verify that $(\overline{U},\overline{Z})$ satisfies the following
BSDE:
\begin{equation}\label{wealthbsde3}
 \overline{U}_t=\xi+\int_t^T \left\{\,-r(u)\overline{U}_u+[\,v(u,c_u)-c_u\,]
 -\max\limits_{\xi\in\Xi}\sum\limits_{j=1}^n\xi_u^jZ_u^j\,\right\}\,du
 -\int_t^T \overline{Z}_u^TdW_u.
\end{equation}
The maximum term in the above bracket can be rewritten in terms of
$\overline{Z}$ as
$$
 \max\limits_{\xi\in\Xi}\sum\limits_{j=1}^n\xi^j_uZ^j_{u}
 =\max\limits_{\xi\in\Xi}\sum\limits_{j=1}^n
 \xi^j_u
 \left(\overline{Z}^j_{u}+\sum\limits_{i=1}^n\sigma_{ij}(u)\pi^i_{u}\right)
 =\delta_{\mathcal{O}}\,(\overline{Z}_u, (\sigma(u))^T\pi_u),
$$
which is Lipschitz continuous in $\overline{Z}_u$, so the comparison
principle holds for (\ref{wealthbsde3}). For any $(\pi,c)\in\Pi$, we
have
$$
v(t,c_t)-c_t\leq v^{*}(t);\qquad
-\delta_{\mathcal{}O}(\overline{Z}_t,(\sigma(t))^T\pi_t)\leq
-d_{\mathcal{O}}(\overline{Z}_t,\mathcal{B}_t),
$$
and for $(\pi,c)=((\sigma^{T})^{-1}{\rm argmin}(\overline{Z},{\cal
B}), \widehat{c}\,)$, we have the equality:
$$
v(t,c_t^*)-c_t^*= v^{*}(t);\qquad
-\delta_{\mathcal{O}}(\overline{Z}_t,{\rm
argmin}(\overline{Z}_t,{\cal B}_t))=
-d_{\mathcal{O}}(\overline{Z}_t,\mathcal{B}_t).
$$
By the BSDE comparison principle, $\overline{U}_t\leq
\overline{U}^*_t$ for any $(\pi,c)\in\Pi$, where $\overline{U}^*$ is
the solution to BSDE:
\begin{equation}\label{wealthbsde30}
 \overline{U}^{*}_t=\xi
 +\int_t^T \Big(\,-r(u)\overline{U}^{*}_u+v^{*}(u)-d_{\mathcal{O}}\,(\overline{Z}^{*}_u,{\cal B}_u)\,\Big)\,du
 -\int_t^T (\overline{Z}^{*}_u)^TdW_u,
\end{equation}
and
$(\overline{U}_t,\overline{Z}_t)=(\overline{U}_t^{*},\overline{Z}_t^{*})$
for $(\pi,c)=((\sigma^{T})^{-1}{\rm argmin}(\overline{Z},{\cal
B}),\hat{c})$.

Therefore,
$$\mathop{{\rm
ess.sup}}\limits_{(\pi,c)\in\Pi[t,T]}U_t(\xi+X^{0;\,\pi,c}_T)=\mathop{{\rm
ess.sup}}\limits_{(\pi,c)\in\Pi[t,T]}\overline{U}_t=\overline{U}^*_t$$
with the optimal portfolio-consumption strategy
$(\pi^*,c^*)=((\sigma^{T})^{-1}{\rm argmin}(\overline{Z}^*,{\cal
B}),\hat{c})$. Finally, it is easy to verify that
$(\overline{U}^*_t-\mathcal{R}(t),\overline{Z}^*_t)=(U^b_t,Z_t^b)$,
which is the unique solution to BSDE (\ref{Pricing5}).

Finally, if $\mathcal{A}=\mathbb{R}^n$, i.e. there is no trading
constraint, $\mathcal{B}_u(\omega)$ is $\mathbb{R}^n$-valued as
well, and
$$d_{\mathcal{O}}(Z^b_u,\mathcal{B}_u(\omega))
=\min_{z\in\mathcal{B}_u(\omega)}\max_{\xi_u(\omega)\in\mathcal{O}}
(\xi_u(\omega))^T(Z_u^b+z)=0$$ where ${\rm
argmin}(Z^b_u,\mathcal{B}_u(\omega))=-Z^b_u$. In this situation, the
pricing BSDE (\ref{Pricing5}) reduces to BSDE
(\ref{contingentbsde}). By the uniqueness of the solution to
(\ref{contingentbsde}), we have $(U^b,Z^b)=(D,Y)$. The optimal
portfolio-consumption strategy in this situation reduces to
$$(\pi^*,c^*)=((\sigma^{T})^{-1}{\rm argmin}(\overline{Z}^*,{\cal
B}),\hat{c})=(-(\sigma^{T})^{-1}Y,\hat{c}).$$
\end{proof}

In the following two sections, we will apply our BSDE representation
results for utility indifference prices to European options and
American options. Another potential application is to consider
exotic options such as Parisian options. For example, Guo et al
\cite{Guo} consider the pricing problem of Parisian options also
under the framework of Chen and Epstein \cite{Chen}. However, they
use the idea of super-replication rather than utility indifference
valuation, so they obtain pricing bounds rather than a price.

\section{Application to European Option}
In this section, we specify our model in a Markovian setting by
assuming the payoff of the contingent claim having the form:
\begin{equation}\label{European}
 \xi=\int_t^T\varrho(u)du+\Psi(S_T),
\end{equation}
where $\varrho$ is the earning rate, and $\Psi$ is the final payoff of the contingent claim at the
maturity $T$. They satisfy the following assumption:
\begin{assumption}\label{Assumption3}
The earnings rate $\varrho(\cdot)$ is a continuous function, and the final payoff $\Psi$ is uniformly Lipschitz continuous:
$$
 |\Psi(S)-\Psi(\overline{S})|\leq K|S-\overline{S}|\ \;\mbox{for}\;S,\,\overline{S}\in\mathbb{R}^n_{+},
$$
so $\Psi(\cdot)$ has linear growth.
\end{assumption}
%\begin{remark}
% From the Lipschitz continuity of $\Psi$, it is clear that
%$$
% |\Psi(S)|\leq |\Psi(1,\cdot\cdot\cdot,1)|+K|S-(1,\cdot\cdot\cdot,1)|\leq C(1+|S|).
%$$
%\end{remark}

Under the above Markovian assumption, the bid price $P^b(t;\xi)$ and
the ask price $P^s(t;\xi)$ can be written as functions of the time
$t$ and the state $S_t$: $P^b(t,S_t;\Psi)\triangleq P^b(t;\xi)$, and
$P^s(t,S_t;\Psi)\triangleq P^s(t;\xi)$. By Theorem \ref{theorem2},
$P^{b}(t,S_t;\Psi)$ and $P^s(t,S_t;\Psi)$ are the solutions to the
following BSDEs respectively: \be\label{FBSDE}
 P^{b}(t,S_t;\Psi)&=&\Psi(S_T)
 +\int_t^T \Big[\,\varrho(u)-r(u)P^{b}(u,S_u;\Psi)-d_{\mathcal{O}}\,(Z^b_u,{\cal B}_u)\,\Big]\,du-\int_t^T
 (Z^b_u)^TdW_u;\;\;
 \\[2mm]\label{FBSDE111}
 P^s(t,S_t;\Psi)&=&\Psi(S_T)
 +\int_t^T \Big[\,\varrho(u)-r(u)P^s(u,S_u;\Psi)+d_{\mathcal{O}}\,(-Z^s_u,{\cal B}_u)\,\Big]\,du-\int_t^T
 (Z^s_u)^TdW_u,\;\;
\ee
where the state $S$ is given by (\ref{stateequation}). Moreover, by
the nonlinear Feynman-Kac formula (see Theorem 4.2 of
\cite{ElKaroui}), $P^b(t,S;\Psi)$ and $P^{s}(t,S;\Psi)$ are the
unique viscosity solutions of the following semi-linear PDEs:
\begin{align}\label{Eouropeanequation}
 \left\{
 \begin{array}{l}
 -\p_t P^b-{\cal L}_0 P^b=\varrho(t)-d_{\mathcal{O}}\,((\sigma(t))^T SD_S P^b,{\cal B}_t)\;\;
 \mbox{in}\;\;{\cal N}_T;
 \vspace{2mm} \\
 -\p_t P^s-{\cal L}_0 P^s=\varrho(t)+d_{\mathcal{O}}\,(-(\sigma(t))^T SD_S P^s,{\cal B}_t)\;\;
 \mbox{in}\;\;{\cal N}_T;
 \vspace{2mm} \\
 P^b(T,S;\Psi)=P^s(T,S;\Psi)=\Psi(S),\quad S\in(0,\,+\infty)^n,
 \end{array}
 \right.
\end{align}
where ${\cal N}_T\triangleq[\,0,T)\times(0,+\infty)^n$ and $SD_S P\triangleq(S_1\p_{S_1} P,\cdots,S_n\p_{S_n} P)^T$, and
the operator $\mathcal{L}_q$ is given by
\be
 {\cal L}_q \triangleq
 \sum_{i,\,j=1}^n{1\over 2}\,a_{ij}(t)\,S_iS_j\p_{S_iS_j}
 +\sum_{i=1}^n\,[\,r(t)-q_i(t)\,]\,S_i\p_{S_i} -r(t)
\ee with
$a_{ij}(t)\triangleq\sum_{l=1}^n\sigma_{il}(t)\,\sigma_{jl}(t)$. The
term $q(\cdot)=(q_1(\cdot),\cdots,q_n(\cdot))^T$ in the operator
$\mathcal{L}_q$ is interpreted as the dividend rate of the risky
assets $S$, and we shall see the bid price and the ask price under
model uncertainty with trading constraint can be related to the
risk-neutral price by adjusting the dividend rate of the underlying.

In order to investigate further properties of the bid and ask prices
and their associated hedging strategies, we need to improve the
regularities of $P^b$ and $P^s$. In the following, we present the
strong solutions for PDEs (\ref{Eouropeanequation}).
\begin{proposition}\label{europeanexistence}
Suppose that Assumptions \ref{Assumption1}, \ref{Assumption2} and
\ref{Assumption3} are satisfied. Then PDEs (\ref{Eouropeanequation})
have unique strong solutions with linear growth. Concretely
speaking,
$$
 P^{b}(t,S;\Psi),\,P^{s}(t,S;\Psi)\in W^{2,\,1}_{p,\,loc}({\cal N}_T)\cap C(\overline{\cal N}_T)\;\;\mbox{for any}\;\;p\geq1,
$$
 and there exists a constant $C$ such that
$$
 |P^{b}(t,S;\Psi)|+|P^{s}(t,S;\Psi)|\leq C(1+|S|)\;\;\mbox{for any}\;\;(t,S)\in\overline{{\cal
 N}}_T,
$$
where $W^{2,\,1}_{p,\,loc}({\cal N}_T)$ is the set of all functions
whose restrictions on the domain ${\cal N}_T^*$ belong to
$W^{2,\,1}_{p}({\cal N}_T^*)$ for any compact subset ${\cal N}_T^*$
of ${\cal N}_T$, and $W^{2,\,1}_{p}({\cal N}_T^*)$ is the completion
of $C^\infty({\cal N}^*_T)$ under the norm:
$$
 \|\,P\,\|\,_{W^{2,\,1}_{p}({\cal N}_T^*)}\triangleq \left[\;\int_ {{\cal N}_T^*}\,
 \left(\,|\,P\,|\,^p+|\partial_tP|^p+|\,D_SP\,|\,^p
 +|\,D^2_SP\,|\,^p\;\right)\,dSdt\,\right]^{1\over p},
 $$
 where $D_SP,\,D_S^2P$ denote the gradient and the Hessian matrix for $P$ with respect to $S$, respectively.
\end{proposition}

We shall show in the Appendix that PDEs (\ref{Eouropeanequation}) is
a special case of the variational inequality (\ref{VI2}) (see
Theorem \ref{connection}). Hence, the above existence and regularity
result is only a special case of the corresponding result for the
variational inequality (\ref{VI2}) in Proposition
\ref{americanexistence}.\\

% solution Problem (\ref{VI2}) can be thought as a special case of Problem (\ref{Eouropeanequation}).
% In fact, if we choose the early payoff function $\Gamma=-C(1+|S|^2)$ and large enough $C$ such that
% $P>\Gamma$ in Problem (\ref{Eouropeanequation}), then Problem (\ref{VI2}) is just Problem (\ref{Eouropeanequation}).
% Then by Proposition \ref{americanexistence} and Assumptions \ref{Assumption1}, \ref{Assumption2} and \ref{Assumption3},
%we conclude that \vspace{1mm}\\

In the rest of this section, we consider a concrete example of the
priors set $\Theta$ by specifying the value set $\mathcal{O}$ of the
corresponding kernel $\xi_t(\omega)\in\mathcal{O}$:
$$
 \mathcal{O}_{1}\triangleq\{x\in\mathbb{R}^n:-\underline\kappa\,_i\leq x_i\leq \overline\kappa_i,\,i=1,\cdots,n\},\qquad
 %\mathcal{O}_2\triangleq\{x\in\mathbb{R}^n:|x|\leq \kappa\},
$$
where $\underline\kappa_i,\overline\kappa_i\geq 0$. The
corresponding priors set is denoted as $\Theta_1$, which is a
generalization of the $\kappa$-ignorance model in Section 3.3 of
\cite{Chen} by taking $\underline\kappa_i=\overline\kappa_i$. With
the above priors set $\Theta_1$, the support function
$\delta_{\mathcal{O}_1}\,(\overline{z},z)$ can be calculated as
$$
 \delta_{\mathcal{O}_1}(\overline{z},z)\triangleq \sum\limits_{i=1}^n\left[\,\overline\kappa_i(\overline{z}_i+z_i)^+
 +\underline\kappa\,_i(\overline{z}_i+z_i)^-\,\right]=\overline{\kappa}^T(\overline{z}+z)^++\underline{\kappa}^T(\overline{z}+z)^-.
$$
 %\delta_{\mathcal{O}_2}(\overline{z},z)&\triangleq \kappa\left[\sum_{i=1}^n(\overline{z}_i+z_i)^2\right]^{\frac12}
% =\kappa|\overline{z}+z|.

By Theorem \ref{theorem2}, if there is no trading constraint, both
the bid price and the ask price coincide with the risk-neutral price
even with model uncertainty. In the following, we specify the
admissible set $\Pi$ by restricting its values in
$\mathcal{A}_1=[0,\infty)^n$, which is equivalent to short sale
constraint. The corresponding admissible set is denoted as $\Pi_1$.

We denote the risk-neutral price under the Black-Scholes model with
the dividend rate $q(\cdot)$ and the payoff $\Psi$ as
$P^0(t,S;q,\Psi)$, and the indifference prices with the priors set
$\Theta_1$ and the admissible set $\Pi_1$ as $P^{1m}(t,S;\Psi)$ for
$m\in\{b,s\}$. We can regard the Black-Scholes framework as a
special case of our indifference pricing model by assuming
$\delta_{\mathcal{O}_0}(\overline{z},z)=0$.\\

Our main results in this section are the connections between the
indifference prices and the risk-neutral prices with different
dividend rates.
%We mainly consider the following two
%cases:
%\begin{itemize}
%%\item (1) The dimension $n=1$, and the final payoff $\Psi(S)$ is
%%monotone in $S$.
%%\item (2) The dimension $n=1$, and the final payoff $\Psi(S)$ is
%%non-monotone in $S$.
%\item Case (1): The final payoff $\Psi(S)$ is co-monotone
%in each component $S^i$, where $S=(S_1,\cdots,S_n)^T$.
%\item Case (2): The final payoff $\Psi(S)$ is non-monotone
%in its components $S=(S_1,\cdots,S_n)^T$.
%\end{itemize}

\begin{proposition} \label{europeaneq}
Suppose that Assumptions \ref{Assumption1}, \ref{Assumption2} and \ref{Assumption3} are
satisfied, the priors set is $\Theta_1$, and the admissible set is $\Pi_1$.
\begin{itemize}
\item If $\Psi(S)$ is increasing in each component $S_i$, then the bid price
is given by $$
P^{1b}(t,S;\Psi)=P^0(t,S;\sigma\overline\kappa,\Psi)$$ with the
optimal portfolio-consumption strategy $(\pi^*,c^*)=(0,\hat{c})$,
and the ask price is given by
$$
P^{1s}(t,S;\Psi)=P^0(t,S;0,\Psi)$$ with the optimal
portfolio-consumption strategy $(\pi^*,c^*)=(SD_SP^0,\hat{c})$.

\item If $\Psi(S)$ is decreasing in each component $S_i$, then the bid price is given by
$$
P^{1b}(t,S;\Psi)=P^0(t,S;0,\Psi)$$ with the optimal
portfolio-consumption strategy $(\pi^*,c^*)=(-SD_SP^0,\hat{c})$, and
the ask price is given by
$$
P^{1s}(t,S;\Psi)=P^0(t,S;\sigma\overline{\kappa},\Psi)$$ with the
optimal portfolio-consumption strategy $(\pi^*,c^*)=(0,\hat{c})$.
\end{itemize}

\end{proposition}

\begin{proof} We only prove the case that $\Psi(S)$ is increasing in each component $S_i$, while the decreasing case is similar.

It is obvious from (\ref{stateequation}) that
$$
 S_T^i=S_t^i\exp\left\{\,\int_t^T\left(\,\,r(u)
 -{1\over2}\,\sum_{j=1}^n|\sigma_{ij}(u)|^2\right)\,du
 +\int_t^T\sum_{j=1}^n\sigma_{ij}(u)dW_u^j\,\right\}.
$$
Hence, if $\Psi(S_T)$ is increasing in each component $S_T^i$, it is
also increasing in $S_t^i$. By the BSDE comparison theorem, for
$m\in\{b,s\}$, $P^{1m}(t,S_t;\Psi)$ is increasing in each component
$S_t^i$ as well. Since $P^{1m}\in W^{2,\,1}_{p,\,loc}({\cal N}_T)$
for any $p\geq 1$, the imbedding theorem for Sobolev space implies
that $D_SP^{1m}\in C({\cal N}_T)$ if we choose $p>n+2$. Hence,
$\p_{S_i}P^{1m}\geq 0$ for each $i=1,2,\cdot\cdot\cdot,n$. Recalling
$\mathcal{B}_t=[0,\infty)^n$, we deduce that for the case of the bid
price: \bee
d_{\mathcal{O}_1}((\sigma(t))^TSD_SP^{1b},\mathcal{B}_t)&=&\min_{z\in\mathcal{B}_t}\left\{\overline\kappa^T((\sigma(t))^TSD_SP^{1b}+z)^+
 +\underline\kappa^T((\sigma(t))^TSD_SP^{1b}+z)^-\right\}\\
 &=&(\sigma(t)\overline{\kappa})^TSD_SP^{1b}
\eee with the optimizer $z^*=0$, or equivalently, $\pi^*=0$. For the
case of the ask price:
$$
d_{\mathcal{O}_1}(-(\sigma(t))^TSD_SP^{1s},\mathcal{B}_t)=\min_{z\in\mathcal{B}_t}\left\{\overline\kappa^T(-(\sigma(t))^TSD_SP^{1s}+z)^+
 +\underline\kappa^T(-(\sigma(t))^TSD_SP^{1s}+z)^-\right\}=0
$$
with the optimizer $z^*=(\sigma(t))^TSD_SP^{1s}$, or
equivalently, $\pi^*=SD_SP^{1s}$. Then the conclusions follow from
the pricing equations (\ref{Eouropeanequation}).
\end{proof}

Intuitively, if the payoff of an option is increasing with the
prices of all the underlying stocks, an investor needs to hold a
short position in each underlying stock in order to hedge a long
position in this option, and a long position in each underlying
stock in order to hedge a short position in this option. However,
since there is short selling constraint (i.e.
$\mathcal{A}_1=[0,\infty)^n$), hedging the long position of the
option is impossible, and the best that the investor can do is not
trading any underlying stocks. In turn, the investor has to
compensate for the option price an equivalent dividend rate
$\sigma\overline{\kappa}$.

The other observation is that the lower bound
$\underline{\kappa}=(\underline{\kappa}_1,\cdots,\underline{\kappa}_1)^T$
in the priors set $\Theta_1$ does not impact the indifference prices
$P^{1b}$ and $P^{1s}$, which is due to the asymmetric property of
the trading constraint set $\mathcal{A}_1=[0,\infty)^n$. Moreover,
the bid-ask spread $P^{1b}-P^{1s}$ is given in terms of the
risk-neutral price with modified dividend rates:
$|P^0(t,S;0,\Psi)-P^0(t,S;\sigma\overline{\kappa},\Psi)|$.

For general payoff $\Psi$, there are no explicit formulae for the
bid price and the ask price. However, we can still have bounds on
the indifference prices in terms of the risk-neutral price with
modified dividend rates. Since PDEs (\ref{Eouropeanequation}) are a
special case of the variational inequality (\ref{VI2}) (see Theorem
\ref{connection}), we will present the proof for the corresponding
variational inequality (\ref{VI2}) in Proposition
\ref{viinequality}, and leave the following proof for PDEs
(\ref{Eouropeanequation}) to the reader.

\begin{proposition} \label{pdeinequality}
Suppose that Assumptions \ref{Assumption1}, \ref{Assumption2} and
\ref{Assumption3} are satisfied, the priors set is $\Theta_1$, and
the admissible set is $\Pi_1$.
\begin{itemize}
\item The bid price satisfies the following inequality:
$$\max\left\{P^0(t,S;\sigma\overline{\kappa},\underline{\Psi}^+),
P^0(t,S;0,\underline{\Psi}^-)\right\}\leq P^{1b}(t,S;\Psi)\leq
P^0(t,S;q,\Psi)$$ for any $q\in[\,0,\sigma\overline{\kappa}\,]$, and the optimal
portfolio-consumption strategy is
$(\pi^*,c^*)=((SD_SP^{1b})^{-},\hat{c})$.
\item The ask price satisfies the following inequality:
$$
P^0(t,S;\sigma\overline{\kappa},\Psi)\leq P^{1s}(t,S;\Psi)\leq
\min\left\{P^0(t,S;\sigma\overline{\kappa},\overline{\Psi}\,^-),
P^0(t,S;0,\overline{\Psi}\,^+)\right\},$$ and the optimal
portfolio-consumption strategy is
$(\pi^*,c^*)=((SD_SP^{1s})^{+},\hat{c})$.
\end{itemize}
Here $\underline{\Psi}^{+}/\underline{\Psi}^{-}$ is any
increasing/decreasing function bounded above by $\Psi$, and
$\overline{\Psi}\,^{+}/\overline{\Psi}\,^{-}$ is any
increasing/decreasing function bounded below by $\Psi$.
\end{proposition}

%\begin{proof} We only prove the case of the bid price, as the case
%of the ask price is similar.
%
%Note that if the final payoff $\Psi$ does not have any monotone
%property, the sign of $D_SP^{1b}$ is indefinite. In this case,
%\begin{align*}
%&\ d_{\mathcal{O}_1}((\sigma(t))^TSD_SP^{1b},\mathcal{B}_t)\\
%=&\ \min_{z\in\mathcal{B}_t}\left\{\overline\kappa^T((\sigma(t))^TSD_SP^{1b}+z)^+
% +\underline\kappa^T((\sigma(t))^TSD_SP^{1b}+z)^-\right\}\\
%=&\ (\sigma(t)\overline{\kappa})^T(SD_SP^{1b})^+
%\end{align*}
%with the optimizer $z^*=(\sigma(t))^T(SD_SP^{1b})^-$, or
%equivalently, $\pi^*=(SD_SP^{1b})^-$. Since
%$-(\sigma(t)\overline{\kappa})^Tx^+\leq
%-(\sigma(t)\overline{\kappa})^Tx$ for any
%$x=(x_1,\cdots,x_n)^T\in\mathbb{R}^d$, by the BSDE comparison
%theorem, we have the upper bound of the bid price:
%$$P^{1b}(t,S_t;\Psi)\leq
%P^0(t,S_t;\sigma\overline{\kappa},\Psi).$$
%
%On the other hand, since $\underline{\Psi}^+/\underline{\Psi}^-\leq
%\Psi$, by the BSDE comparison theorem again, the corresponding bid
%prices are dominated by $P^{1b}(t,S_t;\Psi)$. By Proposition 3.2,
%the bid price associated with the payoff $\underline{\Psi}^+$ is
%$P^0(t,S_t;\sigma\overline{\kappa},\underline{\Psi}^+)$, and the bid
%price associated with the payoff $\underline{\Psi}^-$ is
%$P^0(t,S_t;0,\underline{\Psi}^-)$. Hence, we have the lower bound of
%the bid price:
%$$
%P^{1b}(t,S;\Psi)\geq
%\max\left\{P^0(t,S;\sigma\overline{\kappa},\underline{\Psi}^+),
%P^0(t,S;0,\underline{\Psi}^-)\right\}.$$
%\end{proof}

To finish this section, we investigate how the indifference prices
converge to the corresponding risk-neutral price when the priors set
$\Theta_1$ shrinks to the probability set which only has negative
densities, i.e. the positive part of the model uncertainty
disappears.

\begin{proposition}\label{Europeanconvergence}
Suppose that Assumptions \ref{Assumption1}, \ref{Assumption2} and
\ref{Assumption3} are satisfied, the priors set is $\Theta_1$, and
the admissible set is $\Pi_1$. Then the bid price $P^{1b}$ and the
ask price $P^{1s}$ converge to the risk-neutral price $P^0$ when the
upper bound $\overline{\kappa}$ in the priors set $\Theta_1$
converges to zero. Concretely speaking, \bee
|P^{1b}-P^0|+|P^{1s}-P^0|\leq C\overline{\kappa}^*(1+|S|),\qquad
\Big\|\;|P^{1b}-P^0|+|P^{1s}-P^0|\;\Big\|
 _{W^{2,\,1}_p({\cal N}^*_T)}\leq C_{\mathcal{N}^*_T}\overline{\kappa}^*,
\eee
where $\overline{\kappa}^*=\max\{\overline{\kappa}_1,\cdots,\overline{\kappa}_n\}$,
and $\mathcal{N}^*_T$ is any compact subset of ${\cal N}_T$, and $C$ is a constant independent of $\mathcal{N}^*_T$, but $C_{\mathcal{N}^*_T}$  is a constant depending on $\mathcal{N}^*_T$.
\end{proposition}

We leave its proof in the Appendix.

\section{Application to American Option}

In this section, we extend our model to allow for an early exercise
of the contingent claim. Assume that the contingent claim has the
payoff:
\be\label{American1}
 \xi=\int_t^{\tau\wedge T}\varrho(u)du+\Gamma(\tau,S_\tau)\mathbf{1}_{\{\tau<T\}}+\Psi(S_T)I_{\{\tau=T\}},
\ee where $\tau\in\mathcal{U}[t,T]$ is any $\mathbb{F}$-stopping
time valued in $[t,T],\,\varrho$ is the earning rate, and $\Gamma$
is the early payoff if the option is exercised before the maturity
$T$, and $\Psi$ is the final payoff at the maturity $T$. The earning
rate and the final payoff satisfy Assumption \ref{Assumption3}, and
the early payoff $\Gamma$ satisfies the following assumption:
\begin{assumption}\label{Assumption4}
The early payoff $\Gamma$ is uniformly Lipschitz continuous:
$$
 |\Gamma(t,S)-\Gamma(\overline{t},\overline{S})|\leq
 K(|\,t-\overline{t}\,|)+|\,S-\overline{S}\,|)\;\;\mbox{for}\;t,\,\overline{t}\in[\,0,T\,]\ \text{and}\
 S,\,\overline{S}\in\mathbb{R}^n_+,
$$
and is bounded above by $\Psi$.
\end{assumption}

Since the buyer of the American option has the right to exercise the
option before the maturity $T$, the indifference price in Definition
\ref{def} needs to be modified accordingly. First, note that the
investor's maximum utility satisfies the following time consistency
property: For any $\mathbb{F}$-stopping time
$\tau\in\mathcal{U}[t,T]$,
$$
\mathop{{\rm
ess.sup}}\limits_{(\pi,c)\in\Pi[t,T]}U_t(X^{X_t;\pi,c}_T)=
 \mathop{{\rm ess.sup}}\limits_{(\pi,c)\in\Pi[t,\tau]}U_t\left(\,\mathop{{\rm
 ess.sup}}\limits_{(\pi,c)\in\Pi[\tau,T]}U_{\tau}(X_{T}^{X_t;\pi,c})\,\right)=
\mathop{{\rm
 ess.sup}}\limits_{(\pi,c)\in\Pi[t,\tau]}U_t(\mathcal{R}(\tau)+X_{\tau}^{X_t;\pi,c}),
$$
where we used $(\ref{LHS})$ in the last inequality. In other words,
in order to have the time consistency property, the intermediate
wealth at any $\mathbb{F}$-stopping time $\tau$ consists of not only
the wealth $X_{\tau}^{X_t;\pi,c}$, but also the value of the
remaining optimal consumption $\mathcal{R}(\tau)$ from $\tau$ to the
maturity $T$.

We modify Definition \ref{def}, and give the following definition of
the indifference bid price of the American option.

\begin{definition}\label{definitionofAmerican}

The bid price $P^{b}(t,S_t;\Gamma,\Psi)$ of the American option with
the payoff $\int_t^{\tau\wedge
T}\varrho(u)du+\mathbf{1}_{\{\tau<T\}}\Gamma(\tau,S_{\tau})+
\mathbf{1}_{\{\tau=T\}}\Psi(S_{T})$, where $\tau\in\mathcal{U}[t,T]$
is the exercise time, is defined implicitly by the requirement that
\be\nonumber \mathop{{\rm
 ess.sup}}\limits_{(\pi,c)\in\Pi[t,T]}U_t(X^{X_t;\pi,c}_T)
 &=&\ \mathop{{\rm ess.sup}}\limits_{\tau\in\mathcal{U}[t,T]}\mathop{{\rm ess.sup}}\limits_{(\pi,c)\in\Pi[t,\tau]}
 U_t\Bigg(\int_t^{\tau\wedge T}\varrho(u)du+\Big(\,\mathcal{R}(\tau)+X_{\tau}^{X_t-P^b(t,S_t;\Gamma,\Psi);\pi,c}\\\label{def1}
 &&+\,\Gamma(\tau,S_{\tau})\,\Big)\,\mathbf{1}_{\{\tau<T\}}
 +\Big(\,X_{T}^{X_t-P^b(t,S_t;\Gamma,\Psi);\pi,c}+\Psi(S_{T})\,\Big)\,\mathbf{1}_{\{\tau=T\}}\Bigg).
\ee
\end{definition}

%Let us first consider a simple case $\Gamma(t,S)=\Psi(S)$ for $t\in[0,T]$ and $S\in\mathbb{R}^n_+$, then RHS of (\ref{def1}) reduces to
%$$\mathop{{\rm
% ess.sup}}\limits_{\tau\in[t,T]}
% \mathop{{\rm
% ess.sup}}\limits_{(\pi,c)\in\Pi[t,\tau]}
% U_t\left(\mathcal{R}_{\tau}+X_{\tau}^{X_t-P^b(t,S_t;\Psi);\pi,c}+\Psi(S_{\tau})\right).$$

Similar to the proof of Theorem 2.4, it is easy to check that the
optimization problem on RHS of (\ref{def1}) is translation invariant
of its initial value:
\begin{align*}
&\ \mathop{{\rm
 ess.sup}}\limits_{\tau\in\mathcal{U}[t,T]}
 \mathop{{\rm
 ess.sup}}\limits_{(\pi,c)\in\Pi[t,\tau]}
 U_t\Bigg(\,\int_t^{\tau\wedge T}\varrho(u)du
 +\Big(\,\mathcal{R}(\tau)+X_{\tau}^{0;\pi,c}+\Gamma(\tau,S_{\tau})\,\Big)\,\mathbf{1}_{\{\tau<T\}}\nonumber\\
 &\ +\Big(\,X_{T}^{0;\pi,c}+\Psi(S_{T})\,\Big)\,\mathbf{1}_{\{\tau=T\}}\,\Bigg)+X_t-P^b(t,S_t;\Gamma,\Psi).
\end{align*}

On the other hand, LHS of (\ref{def1}) is $\mathcal{R}(t)+X_t$ by
(\ref{LHS}). Therefore, the bid price of the American option is
\begin{align*}
P^b(t,S_t;\Gamma,\Psi)=&\ \mathop{{\rm ess.sup}}\limits_{\tau\in\mathcal{U}[t,T]}
 \mathop{{\rm ess.sup}}\limits_{(\pi,c)\in\Pi[t,\tau]}
 U_t\Bigg(\,\Big(\,\mathcal{R}(\tau)+X_{\tau}^{0;\pi,c}+\Gamma(\tau,S_{\tau})\,\Big)\,\mathbf{1}_{\{\tau<T\}}\nonumber\\
 &\ +\Big(\,X_{T}^{0;\pi,c}+\Psi(S_{T})\,\Big)\,\mathbf{1}_{\{\tau=T\}}+\int_t^{\tau\wedge T}\varrho(u)du\,\Bigg)-\mathcal{R}(t).
\end{align*}

\begin{theorem}\label{theorem3}
Suppose that Assumptions \ref{Assumption1}, \ref{Assumption2},
\ref{Assumption3} and \ref{Assumption4} are satisfied. Then the bid
price $P^b(t,S_t;\Gamma,\Psi)=U^b_t$, where $U^b_t\geq\Gamma(t,S_t)$
is the unique solution to reflected BSDE:
\begin{align}\label{Pricing7}
 U^{b}_t=\Psi(S_T)
 +\int_t^T \Big[\,\varrho(u)-r(u)U^{b}_u-d_{\mathcal{O}}\,(Z^b_u,{\cal B}_u)\,\Big]\,du+\int_t^TdK_u
 -\int_t^T (Z^b_u)^TdW_u,
 \end{align}
for $(U^b,Z^b)\in L^2_{\mathbb{F}}(0,T;\mathbb{R})\times
L^2_{\mathbb{F}}(0,T;\mathbb{R}^n)$, and $K$ being continuous,
increasing, starting from $K_0=0$, and satisfying the following
Skorohod condition:
$$\int_0^T[\,U^b_u-\Gamma(u,S_u)\,]dK_u=0.$$
The optimal portfolio-consumption strategy for the bid price is
$$(\pi^*,c^*)=((\sigma^{T})^{-1}{\rm argmin}(Z^b,{\cal B}),
\widehat{c}\,),$$ and the optimal exercise time is
$$ \tau^*=\inf\{s\geq t:U^b_s=\Gamma(s,S_s)\}\wedge T.$$
\end{theorem}

\begin{proof}
Similar to the proof of Theorem of \ref{theorem2}, define an
indirect utility $\overline{U}$ by subtracting the wealth
$X^{0;\pi,c}$ from the original utility $U$ as \bee
\overline{U}_{s}&\!\!\!=\!\!\!&\ U_{s}\left(\int_t^{\tau\wedge
T}\varrho(u)du+\Big(\,\mathcal{R}(\tau)
+X_{\tau}^{0;\pi,c}+\Gamma(\tau,S_{\tau})\,\Big)\,\mathbf{1}_{\{\tau<T\}}+\Big(\,X_{T}^{0;\pi,c}+\Psi(S_{T})\,\Big)\,\mathbf{1}_{\{\tau=T\}}\right)-X_{s}^{0;\pi,c};\qquad\\
\overline{Z}_s^j&\!\!\!=\!\!\!&\
Z_s^j-\sum_{i=1}^n\sigma_{ij}(s)\pi^i_{s} \eee for $s\in[t,T]$. By
(\ref{SDE1}) and (\ref{utility}), it is easy to verify that
$(\overline{U},\overline{Z})$ is the solution of the following BSDE:
\begin{align*}
 \overline{U}_t=&\ (\,\mathcal{R}(\tau)+\Gamma(\tau,S_{\tau})\,)\,\mathbf{1}_{\{\tau<T\}}+
 \Psi(S_T)\,\mathbf{1}_{\{\tau=T\}}
 -\int_t^{\tau\wedge T} (\overline{Z}_u)^TdW_u\\[2mm]
 &\ +\int_t^{\tau\wedge T} \left\{\,\varrho(u)-r(u)\overline{U}_u+[\,v(u,c_u)-c_u\,]
 -\delta_{\mathcal{O}}(\overline{Z}_u,(\sigma(u))^T\pi_u)\,\right\}\,du.
\end{align*}

By the BSDE comparison principle, $\overline{U}_t\leq
\overline{U}^*_t$ for any $(\pi,c)\in\Pi[t,\tau]$, where
$\overline{U}^*$ is the solution to BSDE:
\begin{align*}
 \overline{U}_t^*=&\ (\mathcal{R}(\tau)+\Gamma(\tau,S_{\tau}))\mathbf{1}_{\{\tau<T\}}+
 \Psi(S_T)\mathbf{1}_{\{\tau=T\}}\\
 &\ +\int_t^{\tau\wedge T} \left\{\,\varrho(u)-r(u)\overline{U}^*_u+v^*(u)
 -d_{\mathcal{O}}(\overline{Z}_u^*,\mathcal{B}_u)\,\right\}\,du
 -\int_t^{\tau\wedge T} (\overline{Z}_u^*)^TdW_u.
\end{align*}
and
$(\overline{U}_t,\overline{Z}_t)=(\overline{U}_t^{*},\overline{Z}_t^{*})$
for $(\pi,c)=((\sigma^{T})^{-1}{\rm argmin}(\overline{Z},{\cal
B}),\hat{c})$.

Furthermore, by Proposition 2.3 of El Karoui et al \cite{ElKaroui2},
$\overline{U}_t^*\leq \overline{U}_t^{**}$ for any stopping time
$\tau\in\mathcal{U}[t,T]$, where
$\overline{U}^{**}_t\geq\mathcal{R}(t)+\Gamma(t,S_t)$ is the
solution of the following reflected BSDE:
$$
 \overline{U}^{**}_t=\Psi(S_T)
 +\int_t^T \Big[\,\varrho(u)-r(u)\overline{U}^{**}_u+v^*(u)-d_{\mathcal{O}}\,(\overline{Z}^{**}_u,{\cal B}_u)\,\Big]\,du+\int_t^TdK_u
 -\int_t^T (\overline{Z}^{**}_u)^TdW_u,
$$
with the Skorohod condition:
$$\int_0^T[\,\overline{U}^{**}_u-\mathcal{R}(u)-\Gamma(u,S_u)\,]\,dK_u=0,$$
and
$(\overline{U}_t^{*},\overline{Z}_t^{*})=(\overline{U}_t^{**},\overline{Z}_t^{**})$
for $\tau=\tau^*$. Therefore,
\begin{align*}
&\mathop{{\rm
 ess.sup}}\limits_{\tau\in\mathcal{U}[t,T]}
 \mathop{{\rm
 ess.sup}}\limits_{(\pi,c)\in\Pi[t,\tau]}
 U_t\left(\int_t^{\tau\wedge T}\varrho(u)du+(\mathcal{R}(\tau)+X_{\tau}^{0;\pi,c}+\Gamma(\tau,S_{\tau}))\mathbf{1}_{\{\tau<T\}}\right.\\
 &\ \ \ \ \ \ \ \ \ \ \ \ \ \ \ \ \ \ \ \ \ \ \ \ \ \ \ \ \left.+(X_{T}^{0;\pi,c}+\Psi(S_{T}))\mathbf{1}_{\{\tau=T\}}\right)\\
=&\mathop{{\rm
 ess.sup}}\limits_{\tau\in\mathcal{U}[t,T]}
 \mathop{{\rm
 ess.sup}}\limits_{(\pi,c)\in\Pi[t,\tau]}
 \overline{U}_t=\mathop{{\rm
 ess.sup}}\limits_{\tau\in\mathcal{U}[t,T]}\overline{U}_t^{*}=\overline{U}_t^{**}
\end{align*}
with the optimal portfolio-consumption strategy
$(\pi^*,c^*)=((\sigma^{T})^{-1}{\rm argmin}(\overline{Z}^{**},{\cal
B}),\hat{c})$ and the optimal exercise time $\tau^*$.

Finally, it is easy to verify that
$(\overline{U}^{**}_t-\mathcal{R}(t),\overline{Z}^{**}_t)=(U^b_t,Z_t^b)$,
which is the unique solution to the reflected BSDE (\ref{Pricing7}).
\end{proof}

It is important to recall that it is the buyer of the claim who
decides when the contract is exercised. The writer of the derivative
does not have this opportunity and, therefore, she will have to
maximize his utility contingently on the buyer's optimal actions. In
a sense, the valuation problem of the writer reduces to a barrier
type with expiration given by the buyer's optimally chosen
exercise time $\tau^*$, and the payoff is $\mathbf{1}_{\{\tau^*<T\}}\Gamma(\tau,S_{\tau})+
\mathbf{1}_{\{\tau^*=T\}}\Psi(S_{T})$. This asymmetry is not observed in complete
markets where there is a unique price. However, in incomplete
markets such asymmetries naturally emerge and give rise to realistic
price spreads.

\begin{proposition} The ask
price $P^s(t,S_t;\Gamma,\Psi)$ of the American option with the
payoff $\int_t^{\tau^*\wedge
T}\varrho(u)du+\mathbf{1}_{\{\tau^*<T\}}\Gamma(\tau^*,S_{\tau^*})+
\mathbf{1}_{\{\tau^*=T\}}\Psi(S_{T})$ is defined implicitly by the
requirement that
\begin{align}\label{def2}
\mathop{{\rm
 ess.sup}}\limits_{(\pi,c)\in\Pi[t,T]}U_t(X^{X_t;\pi,c}_T)=&\
 \mathop{{\rm
 ess.sup}}\limits_{(\pi,c)\in\Pi[t,\tau^*]}
 U_t\left(\int_t^{\tau^*\wedge
 T}\varrho(u)du+\Big(\mathcal{R}(\tau^*)+X_{\tau^*}^{X_t-P^b(t,S_t;\Gamma,\Psi);\pi,c}\right.\nonumber\\[2mm]
 &\ \left.+\Gamma(\tau^*,S_{\tau^*})\Big)\mathbf{1}_{\{\tau^*<T\}}+\Big(X_{T}^{X_t-P^b(t,S_t;\Gamma,\Psi);\pi,c}+\Psi(S_{T})\Big)\mathbf{1}_{\{\tau^*=T\}}\right).
\end{align}
Suppose that Assumptions \ref{Assumption1}, \ref{Assumption2},
\ref{Assumption3} and \ref{Assumption4} are satisfied. Then the ask
price $P^s(t,S_t;\Gamma,\Psi)=U^s_t$, where $U^s$ is the unique
solution to BSDE:
\begin{align}\label{Pricing8}
U_t^s=&\ \Gamma(\tau^*,S_{\tau^*})\mathbf{1}_{\{\tau^*<T\}}+
 \Psi(S_T)\mathbf{1}_{\{\tau^*=T\}}\nonumber\\
 &\ +\int_t^{\tau^*\wedge T} \Big[\,\varrho(u)-r(u)U^s_u
 +d_{\mathcal{O}}(-Z_u^s,\mathcal{B}_u)\,\Big]\,du
 -\int_t^{\tau^*\wedge T} (Z_u^s)^TdW_u
\end{align}
for $(U^s,Z^s)\in L^2_{\mathbb{F}}(0,T;\mathbb{R})\times
L^2_{\mathbb{F}}(0,T;\mathbb{R}^n)$. The optimal
portfolio-consumption strategy for the ask price is
$$(\pi^*,c^*)=((\sigma^{T})^{-1}{\rm argmin}(-Z^s,{\cal B}\,),
\widehat{c}).$$
\end{proposition}

\begin{proof} The proof is similar to the proof of Theorem \ref{theorem2}, so we omit it.
\end{proof}

By the nonlinear Feynman-Kac formula, $P^b(t,S;\Gamma,\Psi)$ and $P^s(t,S;\Gamma,\Psi)$ are the unique viscosity solutions of the following variational inequality and semi-linear PDE respectively:
\begin{align}\label{VI2}
 \left\{
 \begin{array}{l}
 -\p_t P^b-{\cal L}_0 P^b=\varrho(t)-d_{\mathcal{O}}\,((\sigma(t))^T SD_S P^b,{\cal B}_t)\;\;
 \mbox{if}\;\;P^b>\Gamma\;\mbox{and}\;(t,S)\in{\cal N}_T;
 \vspace{2mm} \\
 -\p_t P^b-{\cal L}_0 P^b\geq \varrho(t)-d_{\mathcal{O}}\,((\sigma(t))^T SD_S P^b,{\cal B}_t)\;\;
 \mbox{if}\;\;P^b=\Gamma\;\mbox{and}\;(t,S)\in{\cal N}_T;
 \vspace{2mm} \\
 P^b(T,S;\Gamma,\Psi)=\Psi(S),
 \end{array}
 \right.
\end{align}
and
\begin{align}\label{VI3}
 \left\{
 \begin{array}{l}
 -\p_t P^s-{\cal L}_0 P^s=\varrho(t)+d_{\mathcal{O}}\,(-(\sigma(t))^T SD_S P^s,{\cal B}_t)\;\;
 \mbox{in}\;\;\{P^b>\Gamma\};
 \vspace{2mm} \\
 P^s=\Gamma\ \ \text{in}\ \{P^b=\Gamma\};\ \ \ \ \ \ P^s(T,S;\Gamma, \Psi)=\Psi(S).
 \end{array}
 \right.
\end{align}

We have the following existence and regularity results for the
strong solutions of (\ref{VI2}) and (\ref{VI3}). Note that the early
payoff function $\Gamma$ usually takes the form
$\Gamma=\max\{\Gamma_1,0\}$, where $\Gamma_1$ is some payoff if the
investor exercises the option. For example, the early payoff
function of American call/put option is just $(S-K)^+/(K-S)^+$ with
the strike price $K$.

\begin{proposition}\label{americanexistence}
Suppose that Assumptions \ref{Assumption1}, \ref{Assumption2}, and
\ref{Assumption3} are satisfied. Moreover, the early payoff function
$\Gamma$ has the form $\Gamma=\max\,\{\,\Gamma_1,\Gamma_2\,\}$,
where both $\Gamma_1$ and $\Gamma_2$ satisfy Assumptions
\ref{Assumption4}, with $D_S^2 \Gamma_1$ and $D_S^2 \Gamma_2$ having
polynomial growth, i.e., there exist a positive constant $C$ and a
positive integer $N$ such that
$$
 |D_S^2 \Gamma_1|+|D_S^2 \Gamma_2|\leq C(1+|S|^N).
$$

Then both (\ref{VI2}) and (\ref{VI3}) have unique strong solutions
with linear growth. Concretely speaking,
\begin{align*}
P^{b}(t,S;\Gamma, \Psi)&\in W^{2,\,1}_{p,\,loc}({\cal N}_T)\cap C(\overline{\cal N}_T);\\
P^{s}(t,S;\Gamma, \Psi)&\in W^{2,\,1}_{p,\,loc}({\cal N}_T\cap\{P^b\geq\Gamma\})\cap C(\overline{\cal N}_T),
\end{align*}
and there exists a constant $C$ such that
$$
 |P^{b}(t,S;\Gamma, \Psi)|+|P^{s}(t,S;\Gamma, \Psi)|\leq C(\,1+|S|\,).
$$
\end{proposition}

We leave its proof in the Appendix.\smallskip\\

In the rest of this section, we consider a concrete example with the
priors set $\Theta_1$ and the admissible set $\Pi_1$, where
$\Theta_1$ and $\Pi_1$ are given in the last section. We denote the
risk-neutral American option price with the dividend rate $q(\cdot)$
and the payoff $\int_t^{\tau\wedge
T}\varrho(u)du+\mathbf{1}_{\{\tau<T\}}\Gamma(\tau,S_{\tau})+
\mathbf{1}_{\{\tau=T\}}\Psi(S_{T})$ as $P^0(t,S;q,\Gamma,\Psi)$, and
the indifference prices with the priors set $\Theta_1$ and the
admissible set $\Pi_1$ as $P^{1m}(t,S;\Gamma,\Psi)$ for
$m\in\{b,s\}$. We can regard the standard American option pricing
framework as a special case of our indifference pricing model for
the bid price by assuming $\delta_{\mathcal{O}_0}(\bar{z},z)=0$.

\begin{proposition} \label{americaneq}
Suppose that the assumptions in Proposition \ref{americanexistence}
are satisfied, the priors set is $\Theta_1$, and the admissible set
is $\Pi_1$.
\begin{itemize}
\item If both $\Gamma(S)$ and $\Psi(S)$ are increasing in each component $S_i$, then the bid price is given by
$$
P^{1b}(t,S;\Gamma,\Psi)=P^0(t,S;\sigma\overline\kappa,\Gamma, \Psi)
$$
with the optimal portfolio-consumption strategy
$(\pi^*,c^*)=(0,\hat{c})$, and the optimal exercise time $$
\tau^*=\widehat{\tau}\triangleq\inf\{s\geq t:P^{1b}(s,S_s;\Gamma,
\Psi)=\Gamma(s,S_s)\}\wedge T.$$
 %$\tau^*$.
 The ask price is bounded by
$$P^{1b}(t,S;\Gamma,\Psi) \leq P^{1s}(t,S;\Gamma, \Psi)\leq P^0(t,S;0,\Gamma, \Psi)$$ with
the optimal portfolio-consumption strategy
$(\pi^*,c^*)=(SD_SP^{1s},\hat{c})$.

\item If both $\Gamma(S)$ and $\Psi(S)$ are decreasing in each component $S_i$, then the bid price is given by
$$
P^{1b}(t,S;\Gamma,\Psi)=P^0(t,S;0,\Gamma,\Psi)$$ with the optimal
portfolio-consumption strategy $(\pi^*,c^*)=(-SD_SP^0,\hat{c})$, and
the optimal exercise time $\tau^*=\widehat{\tau}$. The ask price is
bounded by
$$
P^{1b}(t,S;\Gamma,\Psi)\leq P^{1s}(t,S;\Gamma,\Psi)\leq P^0(t,S;\sigma\overline{\kappa},\Gamma,\Psi)$$
with the
optimal portfolio-consumption strategy $(\pi^*,c^*)=(0,\hat{c})$.
\end{itemize}
\end{proposition}

\begin{proof}
We only prove the case that $\Psi(S)$ is increasing in each
component $S_i$, while the decreasing case is similar. Let
$P^0,\,P^{1s}$ and $P^{1b}$ denote $P^0(t,S;0,\Gamma,\Psi),$
$\,P^{1s}(t,S;\Gamma,\Psi)$ and $P^{1b}(t,S;\Gamma,\Psi)$,
respectively.

First, as in the proof of Proposition \ref{europeaneq}, by the BSDE
comparison theorem and the regularity of the strong solution, we
deduce that $D_{S_i}P^{1m}\geq 0$ for $m\in\{b,s\}$. Therefore, for
the case of the bid price,
\begin{equation*}
d_{\mathcal{O}_1}((\sigma(t))^TSD_SP^{1b},\mathcal{B}_t)
=(\sigma(t)\overline{\kappa})^TSD_SP^{1b}
\end{equation*}
with the optimizer $\pi^*=0$, and for the case
of the ask price:
\begin{equation*}
d_{\mathcal{O}_1}(-(\sigma(t))^TSD_SP^{1s},\mathcal{B}_t)=0
\end{equation*}
with the optimizer $\pi^*=SD_SP^{1s}$. Then
 % the bid price
$P^{1b}(t,S;\Gamma,\Psi)$ $=P^0(t,S;\sigma\overline\kappa,\Gamma,
\Psi) $ follows from the pricing equation (\ref{VI2}).

Different from the European option case, the optimal trading
strategy of the American option's holder affects the seller's ask
price. In other words, the solution of (\ref{VI2}) affects the
solution of (\ref{VI3}), so we can not expect that the ask price is
equal to the price in the standard Black-Scholes market.

We first prove the lower bound of the ask price. It is clear that
$P^{1s}=P^{1b}=\Gamma$ in the domain $\{P^{1b}=\Gamma\}$ from
(\ref{VI3}). In the domain $\{P^{1b}>\Gamma\}$, note that
$d_{\mathcal{O}_1}(\cdot,\mathcal{B}_u)\geq0$, then we have \bee
 -\p_t P^{1b}-{\cal L}_0 P^{1b}&=&\varrho(t)-d_{\mathcal{O}_1}\,((\sigma(t))^T SD_S P^{1b},\mathcal{B}_u)
 \\[2mm]
 &\leq& \varrho(t)+d_{\mathcal{O}_1}\,(-(\sigma(t))^T SD_S P^{1s},\mathcal{B}_u)=-\p_t P^{1s}-{\cal L}_0 P^{1s}.
\eee
The continuity and the terminal and boundary conditions of $P^{1b},\,P^{1s}$ imply $P^{1b}=P^{1s}$ on the parabolic boundary of $\{P^{1b}>\Gamma\}$. By applying Lemma \ref{comparison}, we deduce $P^{1b}\leq P^{1s}$.

Next, we use the similar method to prove the upper bound of the ask price. Since $P^0\geq \Gamma$, we derive that $P^0\geq P^{1s}$ in the domain $\{P^{1b}=\Gamma\}$ from (\ref{VI3}).  Note that $d_{\mathcal{O}_0}(\cdot,\mathcal{B}_u)=0=d_{\mathcal{O}_1}(-(\sigma(t))^TSD_SP^{1s},\mathcal{B}_t)$, we deduce that in the domain $\{P^{1b}>\Gamma\}$
$$
 -\p_t P^{1s}-{\cal L}_0 P^{1s}=\varrho(t)\leq -\p_t P^0-{\cal L}_0 P^0.
$$
 Moreover, the continuity of $P^{1s}$ implies that $P^{1s}=\Gamma\leq P^0$ on $\p_p\{P^{1b}>\Gamma\}\cap{\cal N}_T$.
 On the other hand, $P^{1s}=P^0$ on $\p_p\{P^{1b}>\Gamma\}\cap\p_p {\cal N}_T$. By applying Lemma \ref{comparison}, we deduce $P^{1s}\leq P^0$.
\end{proof}

\begin{proposition}\label{viinequality}
Suppose that the assumptions in Proposition \ref{americanexistence}
are satisfied, the priors set is $\Theta_1$, and the admissible set
is $\Pi_1$. Then we have
\begin{itemize}
\item The bid price satisfies the following inequality:
\begin{align*}
&\ \max\left\{P^0(t,S;\sigma\overline{\kappa},\underline{\Gamma}^+,\underline{\Psi}^+),
P^0(t,S;0,\underline{\Gamma}^-,\underline{\Psi}^-)\right\}\\
\leq&\ P^{1b}(t,S;\Gamma,\Psi)\leq
P^0(t,S;q,\Gamma,\Psi),\;\;\forall\;q\in[\,0,\sigma\overline{\kappa}\,].
\end{align*} The optimal
portfolio-consumption strategy is
$(\pi^*,c^*)=((SD_SP^{1b})^{-},\hat{c})$, and the optimal exercise
time $\tau^*=\widehat{\tau}$.
\item The ask price satisfies the following inequality:
\begin{align*}
&\ P^{1b}(t,S;\Gamma,\Psi)\leq P^{1s}(t,S;\Gamma,\Psi)\\
\leq&\
\min\left\{P^0(t,S;\sigma\overline{\kappa},\overline{\Gamma}\,^-,\overline{\Psi}\,^-),
P^0(t,S;0,\overline{\Gamma}\,^+,\overline{\Psi}\,^+)\right\}.
\end{align*}
The optimal portfolio-consumption strategy is
$(\pi^*,c^*)=((SD_SP^{1s})^{+},\hat{c})$,
\end{itemize}
where $\underline{\Gamma}^{+}/\underline{\Gamma}^{-}$ is any
increasing/decreasing function bounded above by $\Gamma$, and
$\overline{\Gamma}\,^{+}/\overline{\Gamma}\,^{-}$ is any
increasing/decreasing function bounded below by $\Gamma$.
\end{proposition}

%\begin{proof} The proof is similar to the proof of Proposition 3.4, so we omit it.
%\end{proof}

\begin{proof} We only prove the case of the bid price, as the case of the ask price is similar.
Let $P^0$ and $P^{1b}$ denote $P^0(t,S;q,\Gamma,\Psi)$ and
$P^{1b}(t,S;\Gamma,\Psi)$, respectively.

Note that if the final payoff $\Psi$ does not have any monotone
property, the sign of $D_SP^{1b}$ is indefinite. In this case,
\begin{align*}
&\ d_{\mathcal{O}_1}((\sigma(t))^TSD_SP^{1b},\mathcal{B}_t)\\
=&\ \min_{z\in\mathcal{B}_t}\left\{\overline\kappa^T((\sigma(t))^TSD_SP^{1b}+z)^+
 +\underline\kappa^T((\sigma(t))^TSD_SP^{1b}+z)^-\right\}\\
=&\ (\sigma(t)\overline{\kappa})^T(SD_SP^{1b})^+
\end{align*}
with the optimizer $z^*=(\sigma(t))^T(SD_SP^{1b})^-$, or equivalently, $\pi^*=(SD_SP^{1b})^-$. Denote
$$
 \widetilde{\varrho}(t)\triangleq(-\p_t P^{1b}-{\cal L}_q P^{1b})-
 (-\p_t P^{1b}-{\cal L}_0 P^{1b})+\varrho(t)-d_{\mathcal{O}_1}\,((\sigma(t))^T SD_S P^{1b},{\cal B}_t).
$$
Then we can rewrite the variational inequality (\ref{VI2}) as
follows \bee
 \left\{
 \begin{array}{l}
 -\p_t P^{1b}-{\cal L}_q P^{1b}=\widetilde{\varrho}(t)\;\;
 \mbox{if}\;\;P^{1b}>\Gamma\;\mbox{and}\;(t,S)\in{\cal N}_T;
 \vspace{2mm} \\
 -\p_t P^{1b}-{\cal L}_q P^{1b}\geq \widetilde{\varrho}(t)\;\;
 \mbox{if}\;\;P^{1b}=\Gamma\;\mbox{and}\;(t,S)\in{\cal N}_T;
 \vspace{2mm} \\
 P^{1b}(T,S)=\Psi(S).
 \end{array}
 \right.
\eee
 On the other hand, $P^0$ satisfies
\bee
 \left\{
 \begin{array}{l}
 -\p_t P^0-{\cal L}_q P^0=\varrho(t)\;\;
 \mbox{if}\;\;P^0>\Gamma\;\mbox{and}\;(t,S)\in{\cal N}_T;
 \vspace{2mm} \\
 -\p_t P^0-{\cal L}_q P^0\geq \varrho(t)\;\;
 \mbox{if}\;\;P^0=\Gamma\;\mbox{and}\;(t,S)\in{\cal N}_T;
 \vspace{2mm} \\
 P^0(T,S)=\Psi(S),
 \end{array}
 \right.
\eee
 Moreover, it is not difficult to check that
$$
 \widetilde{\varrho}(t)=-\Big[\,(\sigma(t)\overline{\kappa})^T-q(t)^T\,\Big]\,(SD_SP^{1b})^+
 -q(t)^T\Big[\,(SD_SP^{1b})^+-SD_SP^{1b}\,\Big]+\varrho(t)
 \leq \varrho(t).
$$
By applying Lemma \ref{comparison}, we have the upper bound of the
bid price:
$$
 P^{1b}(t,S_t;\Gamma,\Psi)\leq P^0(t,S_t;q,\Gamma,\Psi),\;\;\forall\;q\in[\,0,\sigma\overline{\kappa}\,].
$$

On the other hand, $P^{1b}(t,S;\Gamma,\Psi)$ and
$P^{1b}(t,S_t;\underline{\Gamma}^+/\underline{\Gamma}^-,\,\underline{\Psi}^+/\underline{\Psi}^-)$
satisfy the same differential equation, but with different terminal
values. Lemma \ref{comparison} implies that
$P^{1b}(t,S;\Gamma,\Psi)$ is larger since it has larger terminal
value and obstacle. By Proposition \ref{americaneq}, the bid price
associated with the payoff
$\underline{\Gamma}^+,\,\underline{\Psi}^+$ is
$P^0(t,S_t;\sigma\overline{\kappa},\underline{\Gamma}^+,\underline{\Psi}^+)$,
and the bid price associated with the payoff
$\underline{\Gamma}^-,\,\underline{\Psi}^-$ is
$P^0(t,S_t;0,\underline{\Gamma}^-,\underline{\Psi}^-)$. Hence, we
have the lower bound of the bid price:
$$
P^{1b}(t,S;\Gamma,\Psi)\geq
\max\left\{P^0(t,S_t;\sigma\overline{\kappa},\underline{\Gamma}^+,\underline{\Psi}^+),
P^0(t,S_t;0,\underline{\Gamma}^-,\underline{\Psi}^-)\right\}.$$
\end{proof}

To finish this section, we present a converge result of the
indifference price to the risk-neutral price for the American option
as in Proposition 3.5 for the European option case.

\begin{proposition}\label{americanconvergence}
Suppose the assumptions in Proposition \ref{americanexistence} are
satisfied, the priors set is $\Theta_1$, and the admissible set is
$\Pi_1$. Then the bid price $P^{1b}$ and the ask price $P^{1s}$
converge to the risk-neutral price $P^0$ when the upper bound
$\overline{\kappa}$ in the priors set $\Theta_1$ converges to zero,
Concretely speaking, \bee |P^{1b}-P^0|+|P^{1s}-P^0|\leq
C\overline{\kappa}^*(1+|S|)\;\;\mbox{in}\;{\cal N}_T, \eee where $C$
is a constant independent of $\overline{\kappa}^*$.
\end{proposition}

We leave its proof in the Appendix.

%\section{Numerical Results}
\appendix
\section*{Appendix}
\setcounter{section}{0}
\section{Some Results of Relevant PDEs}

In this Appendix, we provide some technical details on the results
of relevant pricing PDEs for the utility indifference prices. The
main references are Lieberman \cite{Lieberman} and Ladyzenskaja et
al \cite{Ladyzenskaja}.

We consider the following semi-linear parabolic PDE in a general form:
\be\label{semiPDE}
\left\{
 \begin{array}{ll}
 -\p_t P-L P=0\;\;&\mbox{in}\;\;{\cal Q}_T,\vspace{2mm}\\
 P=\Psi\;\;&\mbox{on}\;\;\p_p{\cal Q}_T,
\end{array}
\right.
\ee
where $\p_p{\cal Q}_T$ is the backward parabolic boundary of ${\cal Q}_T$, which is a bounded or unbounded backward parabolic domain, and  the differential operator
$$
 L P\triangleq\sum\limits_{i,j=1}^n a_{ij}\p_{S_iS_j}P +\sum\limits_{i=1}^n b_{i}\p_{S_i}P
 +c P+F(t,S,P,D_S P).
$$

\begin{assumption}\label{Assumption5}
The coefficient function $a$ is continuous in $\overline{{\cal
Q}}_T$, and there exists a positive constant $K$ such that
$$
 |a(t,S)|(1+|S|)^{-2}+|b(t,S)|(1+|S|)^{-1}+|c(t,S)|\leq K\;\;\mbox{for any}\;\;(t,S)\in{\cal Q}_T.
$$
\end{assumption}
\begin{assumption}\label{Assumption6}
There exists a positive constant $K$ such that
$$
 |F(t,S,u_1,v_1)-F(t,S,u_2,v_1)|\leq K(1+|S|)(|u_1-u_2|+|v_1-v_2|)
$$
for any $(t,S)\in{\cal Q}_T,\,u_1,\,u_2\in\mathbb{R},\,v_1,\,v_2\in\mathbb{R}^n$.
\end{assumption}

First, we present the existence and uniqueness result for the strong
solution of (\ref{semiPDE}):

\begin{lemma}\label{existence}
Let ${\cal Q}_T=[\,0,T)\times{\cal Q}$ where ${\cal Q}$ is a bounded
open domain in $\mathbb{R}^n$ with $C^2$ boundary. Suppose that
Assumptions \ref{Assumption5} and \ref{Assumption6} are satisfied,
and $a$ satisfies the uniformly positive definite condition in
${\cal Q}_T$, i.e.,
$$
 \sum\limits_{i,j=1}^na_{ij}(t,S)\xi_i\xi_j\geq |\xi|^2/K\;\;\mbox{for any}\;
 (t,S)\in{\cal Q}_T,\;\xi\in\mathbb{R}^n.
$$
Moreover, $F(\cdot,\cdot,0,0)\in L^p({\cal Q}_T),\,\Psi\in
W^{2,\,1}_p({\cal Q}_T)$ with some $p\geq 1$.

Then  (\ref{semiPDE}) has a unique strong solution $P\in
W^{2,\,1}_p({\cal Q}_T)$. Moreover, the following estimate holds
$$
 \|P\,\|_{W^{2,\,1}_p({\cal Q}_T)}\leq C\Big(\,\|F(\cdot,\cdot,0,0)\|_{L^p({\cal Q}_T)}
 +\|\Psi\|_{W^{2,\,1}_p({\cal Q}_T)}\,\Big),
$$
where the constant $C$ depends on $K,\,p,\,n,\,{\cal Q}_T$.
%but is independent of $P$.
\end{lemma}

Next, we present the interior $W^{2,\,1}_p$ estimate and $C^\alpha$
estimate with local boundary for PDE \eqref{semiPDE}, which is the
key tool to study the problem with low regularity on the boundary,
or the problem in unbounded domain.

\begin{lemma}\label{interiorlp}
Let ${\cal Q}_T$ be a bounded or unbounded backward parabolic
domain, and ${\cal Q}^*_T$ be a compact subset of ${\cal Q}_T$.
Suppose that Assumptions \ref{Assumption5} and \ref{Assumption6} are
satisfied, and $a$ satisfies the uniformly positive definite
condition in ${\cal Q}^*_T$. Moreover, $F(\cdot,\cdot,0,0)\in
L^p({\cal Q}^*_T)$ with some $p\geq 1$.

Then the following estimate holds
$$
 \|P\,\|_{W^{2,\,1}_p({\cal Q}^{**}_T)}\leq C\Big(\,\|F(\cdot,\cdot,0,0)\|_{L^p({\cal Q}^*_T)}
 +\|P\|_{L^p({\cal Q}^*_T)}\,\Big),
$$
where ${\cal Q}^{**}_T$ is a compact subset of ${\cal Q}^*_T$, and
$C$ depends on $K,\,p,\,n,\,{\cal Q}^{**}_T$, and ${\rm dist}({\cal
Q}_T,{\cal Q}^*_T),\,{\rm dist}({\cal Q}^*_T,{\cal Q}^{**}_T)$.
\end{lemma}

\begin{lemma}\label{calpha}
Let ${\cal Q}_T\triangleq [\,0,T)\times{\cal Q}$ where ${\cal Q}$ is
an open domain in $\mathbb{R}^n$ with continuous boundary, and
${\cal Q}^*$ be a compact subset of ${\cal Q}$. Suppose that
Assumptions \ref{Assumption5} and \ref{Assumption6} are satisfied,
and $a$ is uniformly positive definite in ${\cal Q}_T^*\triangleq
[\,0,T)\times{\cal Q}^*$. Moreover, $F(\cdot,\cdot,0,0)\in L^p
({\cal Q}_T^*),\,\Psi\in
C^{\alpha,\,\alpha/2}(\,\overline{{\cal Q}^*}\,)$ with some
$p>n+2,\,\alpha\in(0,1)$.

Then there exists a constant $\beta\in(0,\alpha)$ such that the
following estimate holds
$$
 \|P\,\|_{C^{\beta,\,\beta/2}(\overline{{\cal Q}^{**}_T})}\leq C\Big(\,\|F(\cdot,\cdot,0,0)\|_{L^p({\cal Q}^*_T)}
 +\|P\|_{L^\infty({\cal Q}^*_T)}+\|\Psi\|_{C^{\alpha,\,\alpha/2}(\,\overline{{\cal Q}^*}\;)}\,\Big),
$$
where ${\cal Q}^{**}$ is a compact subset of ${\cal Q}^*$, and
$C,\,\beta$ depend on $K,\,p,\,n$ and $\alpha,\,T,\,{\rm dist}({\cal
Q},{\cal Q}^*),\,{\rm dist}({\cal Q}^*,{\cal Q}^{**})$.
\end{lemma}

We also give the A-B-P comparison principle for the variational
inequality (\ref{VI}) (see Friedman \cite{Friedman2} and Tso
\cite{Tao}), which is similar to the comparison principle for the
classical solution. Since the solutions for variational inequalities
are generally strong solutions rather than classical solutions, the
A-B-P comparison theory is more appropriate for variational
inequalities. On the other hand, PDEs can be regarded as a special
case of variational inequalities if we choose a small enough lower
obstacle $\Gamma$ such that the solution $P>\Gamma$ (see Theorem
\ref{connection}). Hence, the following lemma also applies to PDEs
that we considered.

\begin{lemma}\label{comparison}
For $l=1,2$, let $P^l$  be the strong solutions of the following variational inequalities, respectively
\be\label{VI}
\left\{
 \begin{array}{l}
 -\p_t P^l-L P^l=f^l\;\;\mbox{if}\;\;P^l>\Gamma^l\;\mbox{and}\;(t,x)\in{\cal Q}_T,\vspace{2mm}\\
 -\p_t P^l-L P^l\geq f^l\;\;\mbox{if}\;\;P^l=\Gamma^l\;\mbox{and}\;(t,x)\in{\cal Q}_T,\vspace{2mm}\\
 P^l=\Psi^l\;\;\mbox{on}\;\;\p_p{\cal Q}_T.
\end{array}
\right.
\ee
Suppose that Assumptions \ref{Assumption5} and \ref{Assumption6} are satisfied, and $a$ satisfies the nonnegative definite condition in ${\cal Q}_T$, i.e.,
$$
 \sum\limits_{i,j=1}^na_{ij}(t,S)\xi_i\xi_j\geq0\;\;\mbox{for any}\;\;\xi\in\mathbb{R}^n,\;
 (t,S)\in{\cal Q}_T.
$$
Moreover, $\Psi^l,\,\Gamma^l\in C(\,\overline{{\cal Q}}_T)$, $P^l\in W^{2,\,1}_{p,\,loc}({\cal Q}_T)\cap C(\,\overline{{\cal Q}}_T)$ with some $p>n+2$,
and there exist a positive constant $C$ and a positive integer $N$ such that
$$
 |P^1|+|P^2|\leq C(1+|S|^N).
$$
Then we have $P^1\geq P^2$ in $\overline{{\cal Q}}_T$ if $\,f^1\geq f^2,\;\Psi^1\geq \Psi^2,\;\Gamma^1\geq \Gamma^2\;\;\mbox{in}\;\;{\cal Q}_T$.
\end{lemma}

% \begin{remark}
% In this paper, the boundness condition is not proper for the coefficient functions in the differential operator
% ${\cal L}_q$ in the $t-S$ coordinate system. But we can apply log-transformations (\ref{translation}) to the $t-S$
% coordinate problem and use the corresponding comparison theory in $t-x$ coordinate problem.
% \end{remark}

\section{Proofs of Propositions}

\noindent{\bf Proof of Proposition \ref{Europeanconvergence}.} It is
sufficient to prove the result for $P^{1b}$, and the proof for
$P^{1s}$ is similar. Without loss of generality, we suppose that
$\overline\kappa^*\leq1$.

Note that $d_{\mathcal{O}_1}\,(v,{\cal B}_t)=\overline\kappa^Tv^+$ in this case. Then $P^{1b}$ satisfies
\be\label{europeanconvergenceequation}
 \left\{
 \begin{array}{l}
 -\p_t P^{1b}-{\cal L}_0 P^{1b}=\varrho(t)-(\sigma(t)\overline{\kappa})^T(SD_SP^{1b})^+\;\;\mbox{in}\;\;{\cal N}_T;
 \vspace{2mm}\\
 P^{1b}(T,S)=\Psi(S).
 \end{array}
 \right.
\ee
 By Proposition \ref{pdeinequality}, we have $P^{1b}\leq P^0$. Next, we prove that there exist  positive constants $C_1,\,C_2$ independent of $\overline{\kappa}$ such that
\be\label{convergecelowerbound}
 P^{1b}\geq P^0-\overline{\kappa}^*W,\quad W\triangleq C_1e^{C_2(T-t)}(1+|S|).
\ee

Indeed, since $Q_l=\p_{S_l} P^0,\,l=1,\cdot\cdot\cdot,n$ satisfies
\bee
 \left\{
 \begin{array}{l}
 -\p_t Q_l-\widetilde{{\cal L}}\,^l Q_l=0\;\;\mbox{in}\;{\cal N}_T;
 \vspace{2mm}\\
 Q_l(T,S)=\p_{S_l}\Psi(S),
 \end{array}
 \right.
\eee
where
$$
 \widetilde{{\cal L}}\,^l \triangleq
 \sum_{i,\,j=1}^n{1\over 2}\,a_{ij}(t)\,S_iS_j\p_{S_iS_j}
 +\sum_{i=1}^n\,\Big[\,r(t)+a_{li}(t)\,\Big]\,S_i\p_{S_i}.
$$
By applying Lemma \ref{comparison}, we can deduce that
$|\p_{S_l}P^0|\leq K$, where $K$ is the Lipschitz constant of
$\Psi$.

Moreover, it is not difficult to check that
%\bee
% &&S_i\p_{S_i} W=C_1e^{C_2(T-t)}{S_i^2\over |S|},\qquad |S_i\p_{S_i} W|\leq W;
% \\[2mm]
% &&S_iS_j\p_{S_iS_j} W=C_1e^{C_2(T-t)}\left({\delta_{ij} S_iS_j\over |\,S|}-{S^2_iS^2_j\over|\,S|^3}\right),\quad
% |S_iS_j\p_{S_iS_j} W|\leq W,
%\eee
$$
 S_i\p_{S_i} W=C_1e^{C_2(T-t)}{S_i^2\over |S|},\quad |S_i\p_{S_i} W|\leq W,\quad
 |S_iS_j\p_{S_iS_j} W|\leq W.
$$
%where $\delta_{ij}=1$ for $i=j$ and $\delta_{ij}=0$ for $i\neq j$.
Hence, we have \bee
 &\!\!\!\!\!\!\!\!&-\p_t (P^0-\overline{\kappa}^*W)-{\cal L}_0 (P^0-\overline{\kappa}^*W)-\varrho(t)+(\sigma(t)\overline{\kappa})^T(SD_S(P^0-\overline{\kappa}^*W))^+
 \\[2mm]
 &\!\!\!\!=\!\!\!\!&(-\p_t P^0-{\cal L}_0 P^0)+\overline{\kappa}^*(\p_t W+{\cal L}_0 W)-\varrho(t)+(\sigma(t)\overline{\kappa})^T(SD_SP^0-\overline{\kappa}^*SD_SW)^+\\[2mm]
 &\!\!\!\!\leq\!\!\!\!&-\overline{\kappa}^*(C_2W-CW)+C\overline{\kappa}^*|SD_SP^0|\leq0
\eee provided $C_1,\,C_2$ are large enough. Note that the constants
are independent of $\overline\kappa$.

By applying Lemma \ref{comparison} again, we have proved
\eqref{convergecelowerbound}. By repeating the same argument, we can
obtain $P^{1b}\leq P^0+\overline{\kappa}^*W$. Hence, if we denote
$\Delta P=P^{1b}-P^0$, we have showed that
\be\label{maximumconvergence}
 |\Delta P|\leq \overline\kappa^*\,C_1e^{C_2(T-t)}(1+|S|).
\ee

From PDE~\eqref{europeanconvergenceequation} and Lemma \ref{interiorlp}, we have the following estimate
\begin{equation}\label{grandconvergence}
 \|P^{1b}\,\|_{W^{2,\,1}_p({\cal N}_T^*)}\leq C\Big(\,\|\varrho\|_{L^p({\cal N}_T)}
 +\|P^{1b}\|_{L^p({\cal N}_T)}\,\Big)\leq \overline{C}
\end{equation}
for any compact subset ${\cal N}_T^*$ of ${\cal N}_T$, where the
constants $C,\,\overline C$ depend on ${\cal N}_T^*$, but are
independent of $\overline\kappa$.

%The imbedding theory for Sobolev space implies that
%\be\label{grandconvergence}
% \|D_S P^{1b}\|_{L^p({\cal N}^*_T)}\leq C\|P^{1b}\,\|_{W^{2,\,1}_p({\cal Q}^*_T)}\leq \overline{C}.
%\ee
It is clear that $\Delta P$ satisfies
$$
 -\p_t \Delta P-{\cal L}_0 \Delta P=-(\sigma(t)\overline{\kappa})^T(SD_SP^{1b})^+\;\;\mbox{in}\;\;{\cal N}_T.
$$
By applying Lemma \ref{interiorlp} again, we deduce that
$$
 \|\Delta P\,\|_{W^{2,\,1}_p({\cal Q}^*_T)}\leq C\Big(\,\|(\sigma(t)\overline{\kappa})^T(SD_SP^{1b})^+\|_{L^p({\cal Q}^*_T)}
 +\|\Delta P\|_{L^p({\cal Q}^*_T)}\,\Big)\leq \overline{C}\,\overline\kappa^*,
$$ where we have used~\eqref{maximumconvergence} and
\eqref{grandconvergence}.\medskip\\

\noindent{\bf Proof of Proposition \ref{americanexistence}.} First,
we prove the existence of the strong solution $P^b$ for the
variational inequality (\ref{VI2}).

% In order to remove the degeneracy of the equation, we apply the log-transformations of (\ref{VI2})
% equations:
% \be\label{translation}
%  x_i=\ln S_i,\quad p\,^b(t,x;\Gamma,\Psi)=P^b(t,S;\Gamma,\Psi).
% \ee
%  Then $p^b$ satisfies
% \be\label{nodegenerator}
%  \left\{
%  \begin{array}{l}
%  -\p_t p\,^{b}-\widetilde{{\cal L}} p\,^{b}=\varrho(t)-d(\sigma(t) D_x p\,^{b},{\cal B}_u)\;\;\;\;\;
%  \mbox{if}\;p\,^{b}>\gamma\;\mbox{and}\;(t,x)\in{\cal M}_T;
%  \vspace{2mm} \\
%  -\p_t p\,^{b}-\widetilde{{\cal L}} p\,^{b}\geq \varrho(t)-d(\sigma(t) D_x p\,^{b},{\cal B}_u)\;\;\;\;\;
%  \mbox{if}\;p\,^{b}=\gamma\;\mbox{and}\;(t,x)\in{\cal M}_T;
%  \vspace{2mm} \\
%  p\,^{b}(T,x)=\Psi(e^x),\quad x\in\mathbb{R}^n,
%  \end{array}
%  \right.
% \ee
%  where $\gamma(t,x)=\Gamma(t,S),\,\psi(x)=\Psi(S)$, and
% $$
%  \widetilde{{\cal L}} p\triangleq
%  \sum_{i,\,j=1}^n{1\over 2}\,a_{ij}(t)\,\p_{x_ix_j}p
%  +\sum_{i=1}^n\,\Big[\,r(t)-{1\over 2}\,a_{ii}(t)\,\Big]\,\p_{x_i} p-r(t)p.
% $$

We use the penalty method to approximate the variational
inequality~\eqref{VI2}. \be\label{penalty}
 \left\{
 \begin{array}{l}
 -\p_t P_m-{\cal L}_0 P_m=\varrho(t)-d_{\mathcal{O}}\,((\sigma(t))^T SD_S P_m,{\cal B}_t)
 +m(P_m-\Gamma)^-\;\mbox{in}\;{\cal N}_T;
 \vspace{2mm} \\
 P_m(T,S)=\Psi(S).
 \end{array}
 \right.
\ee

Since the above problem lies in unbounded domain, and the regularity
of the terminal value $\Psi$ is not enough, we need to smooth $\Psi$
and use the following problem in bounded domain to approximate the
above PDE in unbounded domain.
 \be\label{boundedproblem}
 \left\{
 \begin{array}{l}
 -\p_t P_{k,\,m}-{\cal L}_0 P_{k,\,m}=\varrho(t)-d_{\mathcal{O}}\,((\sigma(t))^T SD_S P_{k,\,m},{\cal B}_t)
 +m(P_{k,\,m}-\Gamma)^-\;\mbox{in}\;{\cal N}^k_T;
 \vspace{2mm} \\
 P_{k,\,m}(t,S)=\Psi_k(S)\;\;\mbox{on}\;\p_p {\cal N}^k_T,
 \end{array}
 \right.
\ee where ${\cal N}^k_T\triangleq [\,0,T)\times {\cal N}^k$, and
${\cal N}^k\triangleq \{S\in \mathbb{R}^n:1/k\leq S_i\leq
k,\,i=1,\cdot\cdot\cdot,n\},\,k\in\mathbb{N}_+$. $\Psi_k$ is the
smooth function of $\Psi$, which is defined as follows.  Denote by
$\eta$ the standard mollifier, then $\eta_k(S)=k^n\eta(kS)$, and
$$
\Psi_k(S)=\int_{\mathbb{R}^n}\Psi(y)\,\eta_k(S-y)\,dy+{K\over
k}\;\;\mbox{with}\;\; \Psi(y)=0\;\mbox{if}\;y\notin{\cal N}.
$$

It is clear that $\Psi_k\in C^\infty(\overline{{\cal Q}})$ for any
$k\in \mathbb{N}_+$. It is not difficult to deduce that $\Psi_k$
converges to $\Psi$ in $C^{\alpha}(\overline{{\cal N}^R})$ for any
$\alpha\in(0,1),\,R\in\mathbb{N}_+$ as $k\rightarrow\infty$, and
there exists a constant $C$ independent of $m,\,k$ and $\Gamma$ such
that \be\label{uniformestimate}
 |\Psi_k|\leq C(1+|S|),\qquad |D_S\Psi_k| \leq C,\qquad \Psi_k\geq\Gamma.
\ee

Since the nonlinear term $-d_{\cal O}(\sigma(t)v,{\cal
B}_t)+m(u-\Gamma)^-$ is Lipschitz continuous with respect to $u,v$,
and $-d_{\cal O}(0,{\cal B}_t)+m(0-\Gamma)^-$ is bounded in ${\cal
N}^k_T$, we deduce that problem~\eqref{boundedproblem} has a unique
strong solution $P_{k,\,m}\in W^{2,\,1}_p({\cal N}^k_T)\cap
C(\,\overline{{\cal N}^k_T}\,)$ by Lemma \ref{existence}.

Next, we prove some estimates for $P_{k,m}$ which are independent of
$k,\,m$, in order to achieve some proper convergence results. First,
we show the upper bound of $P_{k,\,m}$. More precisely, there exist
constants $C_1,\,C_2$ independent of $m,\,k$ such that
\be\label{upperbound}
 P_{k,\,m}\leq W\triangleq C_1e^{C_2(T-t)}(1+|S|).
\ee

Indeed, we first choose $C_1$ and $C_2$ large enough such that
$W\geq \Gamma$, then we have
$$-\p_t W-{\cal L}_0
W-\varrho(t)+d_{\mathcal{O}}\,((\sigma(t))^T SD_S W,{\cal
B}_t)-m(W-\Gamma)^- \geq C_2W-CW-C\geq0. $$ Moreover,
by~\eqref{uniformestimate}, it is clear that
$$
 W(t,S)\geq \Psi_k(S)=P_{k,\,m}(t,S)\;\mbox{on}\;\p_p {\cal N}^k_T
$$
provided $C_1,\,C_2$ are large enough. Then Lemma \ref{comparison}
implies (\ref{upperbound}).

Next, we prove the lower bound of $P_{k,\,m}$, i.e. there exist
constants $C_3,\,C_4$ independent of $m,\,k$ and $\Gamma$ such that
\be\label{lowerbound}
 P_{k,\,m}\geq w\triangleq -C_3e^{C_4(T-t)}(1+|S|).
\ee Indeed, we can choose $C_3$ and $C_4$ large enough such that
\bee
 &&-\p_t w-{\cal L}_0 w-\varrho(t)+d_{\mathcal{O}}\,((\sigma(t))^T SD_S w,{\cal B}_t)-m(w-\Gamma)^-
 \leq C_4w-Cw+C+C|(\sigma(t))^T SD_S w|
 \leq0;\\[2mm]
 &&w(t,S)\leq \Psi_k(S)=P_{k,\,m}(t,S)\;\mbox{on}\;\p_p {\cal N}^k_T.
\eee
Then Lemma \ref{comparison} implies (\ref{lowerbound}).

By applying Lemma \ref{interiorlp} and \ref{calpha} to
PDE~\eqref{boundedproblem} in the domain ${\cal N}^R_T$ with
$k>R,\,R\in \mathbb{N}_+$, we derive that there exists a constant
$C$ depending on $m,\,R$, but is independent of $k$ such that \bee
 &&\|P_{k,\,m}\|_{W^{2,\,1}_p({\cal N}^R_T\cap\{t\leq T-1/R\})}
 +\|P_{k,\,m}\|_{C^{\beta,\,\beta/2}(\overline{{\cal N}^R_T})}
 \\[2mm]
&\leq& C_R\Big(\|P_{k,\,m}\|_{L^\infty({\cal N}^R_T)}
 +\|d_{\mathcal{O}}\,(0,{\cal B}_t)\|_{L^p({\cal N}^R_T)}
 +\|m(P_{k,\,m}-\Gamma)^-\|_{L^p({\cal N}^R_T)}
 +\|\Psi_k\|_{C^{\alpha}(\overline{{\cal N}^R})}+1\Big)\leq C,
\eee where we have used~\eqref{upperbound} and \eqref{lowerbound}.

Due to the above estimates, we can show that the solution of
\eqref{boundedproblem} approximates the solution to \eqref{penalty}
by the method in \cite{Yang,Yang1}. More precisely, there exists a
function $P_m$ such that some subsequence of
$\{P_{k,\,m}\}_{k=1}^\infty$ converges to $P_m$ weakly in
$W^{2,\,1}_p({\cal N}^R_T\cap\{t\leq T-1/R\})$ and strongly in
$C(\,\overline{{\cal N}^R_T}\,)$ for any $m,R\in\mathbb{N}_+$, and
$P_m$ is the strong solution of \eqref{penalty}. Moreover, by taking
$k\rightarrow\infty$ in \eqref{upperbound} and \eqref{lowerbound},
we have \be\label{boundmaximum}
 -C_3e^{C_4(T-t)}(1+|S|)\leq P_m\leq C_1e^{C_2(T-t)}(1+|S|)\;\;\mbox{in}\;\overline{{\cal N}_T},
\ee where $C_1,\,C_2,\,C_3,\,C_4$ are independent of $m$, and
$C_3,\,C_4$ are independent of $\Gamma$.

In order to show that the solution of \eqref{penalty} approximates
the solution of the variational inequality \eqref{VI2}, we need to
prove another lower bound of $P_m$ such that \be\label{lowerbound2}
 P_m\geq \Gamma -{W\over m},\quad W \triangleq C_1e^{C_2(T-t)}(1+|S|^{N+2})
\ee provided $m$ is large enough, and where $C_1,\,C_2$ are
constants independent of $m$. Indeed, denote $w^*=\Gamma_1 -W/ m$,
and we can check that \bee &&-\p_t w^*-{\cal L}_0
w^*-\varrho(t)+d_{\mathcal{O}}\,((\sigma(t))^T SD_S w^*,{\cal
B}_t)-m(w^*-\Gamma)^-
\\[2mm]
&\leq& (-\p_t \Gamma_1-{\cal L}_0 \Gamma_1)-{C_2W-CW\over
m}+C+C\,|(\sigma(t))^T SD_S\Gamma_1|-W\leq C(1+|S|^{N+2})-W\leq0
\eee provided $C_1,\,C_2$ are large enough. Then Lemma
\ref{comparison} implies $P_m\geq \Gamma_1 -W/ m$. By repeating the
same argument, we can deduce that $P_{k,\,m}\geq \Gamma_2 -W/ m$, so
\eqref{lowerbound2} is obvious.

By applying Lemma \ref{interiorlp} and \ref{calpha}, we derive that
there exists a constant $C_R$ independent of $m$ such that \bee
 &&\|P_m\|_{W^{2,\,1}_p({\cal N}^R_T\cap\{t\leq T-1/R\})}
 +\|P_m\|_{C^{\beta,\,\beta/2}(\overline{{\cal N}^R_T})}
 \\[2mm]
 &\leq&C_R\Big(\|P_m\|_{L^\infty({\cal N}^R_T)}
 +\|d_{\mathcal{O}}\,(0,{\cal B}_t)\|_{L^p({\cal N}^R_T)}+\|m(P_m-\Gamma)^-\|_{L^p({\cal N}^R_T)}
 +\|\Psi\|_{C^{\alpha}(\overline{{\cal N}^R})}+1\Big)\leq C,
\eee where we have used~\eqref{boundmaximum} and
\eqref{lowerbound2}, and the constant $C$ is independent of $m$.

Thanks to the above estimates, we can show that the solution of
 \eqref{penalty} approximates the solution of \eqref{VI2} by the method in
\cite{Friedman2,Yang}. More precisely, there exists a function $P$
such that some subsequence of $\{P_m\}_{m=1}^\infty$ converges to
$P$ weakly in $W^{2,\,1}_p({\cal N}^R_T\cap\{t\leq T-1/R\})$ and
strongly in $C(\,\overline{{\cal N}^R_T}\,)$ for any
$R\in\mathbb{N}_+$, and $P$ is the strong solution of
Problem~\eqref{VI2}.  Therefore, we have proved the existence of the
strong solution $P^b$ for the variational inequality (\ref{VI2}).

Since $R$ is arbitrary, we deduce that $P^b\in
W^{2,\,1}_{p,\,loc}({\cal N}_T)\cap C(\overline{{\cal N}_T})$.
Moreover, by taking $m\rightarrow\infty$ in \eqref{boundmaximum}, we
have \be\label{bound}
 -C_3e^{C_4(T-t)}(1+|S|)\leq P^b\leq C_1e^{C_2(T-t)}(1+|S|)\;\;\mbox{in}\;\overline{{\cal N}_T},
\ee where $C_3,\,C_4$ are independent of $\Gamma$. The uniqueness of
the strong solution $P^b$ for the variational inequality (\ref{VI2})
follows from Lemma \ref{comparison}. The proof of the results
for $P^s$ is similar.\\

\noindent{\bf Proof of Proposition \ref{americanconvergence}.} The
proof is similar to that of Proposition  \ref{Europeanconvergence}.
First, we prove  the result for $P^{1b}$. Without loss of
generality, we suppose that $\overline\kappa^*\leq1$.

We first prove that \be\label{firstderivative2}
 |\p_{S_l} P^0|\leq K\;\;\;\mbox{for any}\;\;l=1,\cdot\cdot\cdot,n,
\ee where $K$ is the Lipschitz constant of $\Gamma$ and $\Psi$.

We must start from the penalty problem~\eqref{penalty} to prove the
result because the  regularity of the solution for the variational
inequality~\eqref{VI2} is not enough. Note that
$d_{\mathcal{O}_0}\,(v,{\cal B}_t)=0$ in this case, then the penalty
problem is \bee
 \left\{
 \begin{array}{l}
 -\p_t P^0_m-{\cal L}_0 P^0_m=\varrho(t)+m(P^0_m-\Gamma)^-\;;
 \vspace{2mm} \\
 P^0_m(T,S)=\Psi(S).
 \end{array}
 \right.
\eee

Denote $Q_l=\p_{S_l} P^0_m,\,l=1,\cdot\cdot\cdot,n$ which satisfies
\bee
 \left\{
 \begin{array}{l}
 -\p_t Q_l-\widehat{\cal L}\,^l Q_l=m\p_{S_l}\Gamma \,I_{\{P^0_m<\Gamma\}}\;\;\mbox{in}\;{\cal N}_T;
 \vspace{2mm}\\
 Q_l(T,S)=\p_{S_l}\Psi(S),
 \end{array}
 \right.
\eee
where
$$
 \widehat{{\cal L}}\,^l \triangleq
 \sum_{i,\,j=1}^n{1\over 2}\,a_{ij}(t)\,S_iS_j\p_{S_iS_j}
 +\sum_{i=1}^n\,\Big[\,r(t)+a_{li}(t)\,\Big]\,S_i\p_{S_i}
 -mI_{\{P^0_m<\Gamma\}}.
$$

By applying Lemma \ref{comparison}, we have
$$
 |\p_{S_l} P^0_m|\leq K\;\;\;\mbox{for any}\;\;l=1,\cdot\cdot\cdot,n.
$$
By letting $m\rightarrow\infty$, we deduce~\eqref{firstderivative2}.

Note that $d_{\mathcal{O}_1}\,(v,{\cal B}_t)=\overline\kappa^Tv^+$ in this case. Then $P^{1b}$ satisfies
\be\label{americanconvergenceequation}
 \left\{
 \begin{array}{l}
 -\p_t P^{1b}-{\cal L}_0 P^{1b}=\varrho(t)-(\sigma(t)\overline{\kappa})^T(SD_SP^{1b})^+\;\;\mbox{if}\;\;P^{1b}>\Gamma;
 \vspace{2mm}\\
 -\p_t P^{1b}-{\cal L}_0 P^{1b}\geq \varrho(t)-(\sigma(t)\overline{\kappa})^T(SD_SP^{1b})^+\;\;\mbox{if}\;\;P^{1b}=\Gamma;
 \vspace{2mm}\\
 P^{1b}(T,S)=\Psi(S),
 \end{array}
 \right.
\ee

From Proposition \ref{viinequality}, we have $P^{1b}\leq P^0$. Next,
we prove that there exist  positive constants $C_1,\,C_2$
independent of $\overline{\kappa}$ such that
\be\label{convergecelowerbound2}
 P^{1b}\geq w^*=P^0-\overline{\kappa}^*W,\quad W\triangleq C_1e^{C_2(T-t)}(1+|S|).
\ee

Indeed, repeating the same argument as in the proof of Proposition
\ref{Europeanconvergence}, we have \bee
 &\!\!\!\!\!\!\!\!&\widetilde{\varrho}(t)\triangleq\Big(\,-\p_t w^*-{\cal L}_0 w^*
 +(\sigma(t)\overline{\kappa})^T(SD_Sw^*)^+\,\Big)-\Big(\,-\p_t P^0-{\cal L}_0 P^0\,\Big)
 +\varrho(t)
 \\[2mm]
 &\!\!\!\!=\!\!\!\!&\overline{\kappa}^*(\p_t W+{\cal L}_0 W)+(\sigma(t)\overline{\kappa})^T(SD_SP^0-\overline{\kappa}^*SD_SW)^++\varrho(t)
 \\[2mm]
 &\!\!\!\!\leq\!\!\!\!&-\overline{\kappa}^*(C_2W-CW)+C\overline{\kappa}^*|SD_SP^0|+\varrho(t)
 \leq \varrho(t)
\eee provided $C_1,\,C_2$ are large enough, where we have
used~\eqref{firstderivative2}.

Hence, the variational inequality for $P^0$ implies that $w^*$
satisfies \bee
 \left\{
 \begin{array}{l}
 -\p_t w^*-{\cal L}_0 w^*+(\sigma(t)\overline{\kappa})^T(SD_Sw^*)^+=\widetilde{\varrho}(t)\;\;
 \mbox{if}\;\;w^*>\Gamma-\overline{\kappa}^*W;
 \vspace{2mm} \\
 -\p_t w^*-{\cal L}_0 w^*+(\sigma(t)\overline{\kappa})^T(SD_Sw^*)^+\geq \widetilde{\varrho}(t)\;\;
 \mbox{if}\;\;w^*=\Gamma-\overline{\kappa}^*W;
 \vspace{2mm} \\
 w^*(T,S)=\Psi(S)-\overline{\kappa}^*W(T,S).
 \end{array}
 \right.
\eee

By applying Lemma \ref{comparison} again, we have
\eqref{convergecelowerbound2}, so the result for $P^{1b}$ has been
proved.

Next, we prove the result for $P^{1s}$. Since $P^{1s}=\Gamma=P^{1b}$
in the set $\{P^{1b}=\Gamma\}$, the result for $P^{1b}$ implies that
$$
|P^{1s}-P^0|\leq C\overline\kappa^*\;\;\mbox{in}\;\;\{P^{1b}=\Gamma\}.
$$
Then it is sufficient to prove that the above inequality holds in
the set $\{P^{1b}>\Gamma\}$.

From Proposition \ref{viinequality}, we have $P^{1b}\leq P^0$, so
that
$$
\{P^{1b}>\Gamma\}\subset\{P^0>\Gamma\}.
$$

Hence, in the set $\{P^{1b}>\Gamma\}$, $P^{1s}$ and $P^0$ satisfy
\bee
 \left\{
 \begin{array}{l}
 -\p_t P^{1s}-{\cal L}_0 P^{1s}=\varrho(t)+(\sigma(t)\overline{\kappa})^T(SD_SP^{1s})^-\;\;\mbox{in}\;\{P^{1b}>\Gamma\};
 \vspace{2mm}\\
 -\p_t P^{0}-{\cal L}_0 P^{0}=\varrho(t)\;\;\mbox{in}\;\{P^{1b}>\Gamma\};
 \vspace{2mm}\\
 | P^{1s}-P^{0}|=| P^{1b}-P^{0}|\leq C\overline\kappa^*(1+|S|)\;\;\mbox{on}\;\p_p\{P^{1b}>\Gamma\}.
 \end{array}
 \right.
\eee Repeating the same argument as above, we have proved the result
for
$P^{1s}$.\\

To finish the Appendix, we give the following connection between the
variational inequality \eqref{VI2} and PDEs
(\ref{Eouropeanequation}), which is a direct consequence of
(\ref{bound}).

\begin{theorem} \label{connection}
Suppose that Assumptions \ref{Assumption1}, \ref{Assumption2},
\ref{Assumption3} are satisfied. Then there exist some functions
$\Gamma\in C^\infty({\cal N}_T)$ satisfying Assumption
\ref{Assumption4}, such that \eqref{VI2} is equivalent to
(\ref{Eouropeanequation}).
\end{theorem}
\begin{proof}
 From~\eqref{bound} and Assumption \ref{Assumption3}, we can choose
 a large enough constant $C$ such that
$$
 \Gamma=-2C\sqrt{1+|S|^2},
$$
which satisfies Assumption \ref{Assumption4}, and $\Gamma\in
C^\infty({\cal N}_T)$ with $P^b>\Gamma$. Since $P^b>\Gamma$, then
\eqref{VI2} implies that
$$
 -\p_t P^b-{\cal L}_0 P^b=\varrho(t)-d_{\mathcal{O}}\,((\sigma(t))^T SD_S P^b,{\cal B}_t)\;\;\mbox{in}\;\;{\cal N}_T.
$$
Hence, \eqref{VI2} is  equivalent to (\ref{Eouropeanequation}).
\end{proof}

\end{document}